\def\isarxivversion{1} 

\ifdefined\isarxivversion
\documentclass[11pt]{article}
\else
\documentclass{article}
\usepackage[final]{neurips_2019}
\fi

\usepackage{amsmath}
\usepackage{amsthm}
\usepackage{amssymb}
\usepackage{algorithm}
\usepackage{subfig}
\usepackage{color}
\usepackage[english]{babel}
\usepackage{graphicx}
\usepackage{wrapfig,epsfig}
\usepackage{epstopdf}
\usepackage{url}
\usepackage{graphicx}
\usepackage{color}
\usepackage{epstopdf}
\usepackage{algpseudocode}
\usepackage{scrextend}
\usepackage[T1]{fontenc}
\usepackage{bbm}
\usepackage{comment}

\usepackage{tikz}
\usepackage{hyperref}  
\hypersetup{colorlinks=true,citecolor=blue,linkcolor=blue} 
\usetikzlibrary{arrows}

\graphicspath{{./figs/}}

\ifdefined\isarxivversion
\usepackage[margin=1in]{geometry}
\else 

\fi

\ifdefined\isarxivversion
\author{
	Zhao Song\thanks{\texttt{zhaosong@uw.edu}. University of Washington.}
	\and
	David P. Woodruff\thanks{\texttt{dwoodruf@cs.cmu.edu}. Carnegie Mellon University.}
	\and
	Peilin Zhong\thanks{\texttt{pz2225@columbia.edu}. Columbia University.}
}
\else

\author{
	Zhao Song\thanks{equal contribution.}\\
	University of Washington \\
	\texttt{magic.linuxkde@gmail.com}
	\And
	David P. Woodruff\printfnsymbol{1}  \\
	Carnegie Mellon University\\
	\texttt{dwoodruf@cs.cmu.edu}
	\And
	Peilin Zhong\printfnsymbol{1} \\
	Columbia University\\
	\texttt{pz2225@columbia.edu}
}

\fi

\date{}
\ifdefined\isarxivversion
\title{Average Case Column Subset Selection for\\ Entrywise $\ell_1$-Norm Loss\thanks{A preliminary version of this paper appears in Proceedings of Thirty-third Conference on Neural Information Processing Systems (NeurIPS 2019).}}
\else
\title{Average Case Column Subset Selection for Entrywise $\ell_1$-Norm Loss}
\fi

\newtheorem{theorem}{Theorem}[section]
\newtheorem{lemma}[theorem]{Lemma}
\newtheorem{definition}[theorem]{Definition}

\newtheorem{fact}[theorem]{Fact}

\newtheorem{claim}[theorem]{Claim}

\newcommand{\wh}{\widehat}
\newcommand{\wt}{\widetilde}

\newcommand{\R}{\mathbb{R}}

\renewcommand{\varepsilon}{\epsilon}
\renewcommand{\tilde}{\wt}
\renewcommand{\hat}{\wh}

\DeclareMathOperator*{\E}{{\bf {E}}}

\DeclareMathOperator{\poly}{poly}

\DeclareMathOperator{\rank}{rank}

\DeclareMathOperator{\sign}{sign}

\definecolor{mygreen}{RGB}{80,180,0}
\definecolor{b2}{RGB}{51,153,255}
\definecolor{mycy2}{RGB}{255,51,255}

\makeatletter
\newcommand*{\RN}[1]{\expandafter\@slowromancap\romannumeral #1@}
\newcommand{\printfnsymbol}[1]{%
	\textsuperscript{\@fnsymbol{#1}}%
}
\makeatother

\begin{document}

\ifdefined\isarxivversion

\begin{titlepage}
\maketitle
\begin{abstract}
We study the column subset selection problem with respect to the entrywise $\ell_1$-norm loss. It is known that in the worst case, to obtain a good rank-$k$ approximation to a matrix, one needs an arbitrarily large $n^{\Omega(1)}$ number of columns to obtain a $(1+\epsilon)$-approximation to the best entrywise $\ell_1$-norm low rank approximation of an $n \times n$ matrix. Nevertheless, we show that under certain minimal and realistic distributional settings, it is possible to obtain a $(1+\epsilon)$-approximation with a nearly linear running time and poly$(k/\epsilon)+O(k\log n)$ columns. Namely, we show that if the input matrix $A$ has the form $A = B + E$, where $B$ is an arbitrary rank-$k$ matrix, and $E$ is a matrix with i.i.d. entries drawn from any distribution $\mu$ for which the $(1+\gamma)$-th moment exists, for an arbitrarily small constant $\gamma > 0$, then it is possible to obtain a $(1+\epsilon)$-approximate column subset selection to the entrywise $\ell_1$-norm in nearly linear time. Conversely we show that if the first moment does not exist, then it is not possible to obtain a $(1+\epsilon)$-approximate subset selection algorithm even if one chooses any $n^{o(1)}$ columns. This is the first algorithm of any kind for achieving a $(1+\epsilon)$-approximation for entrywise $\ell_1$-norm loss low rank approximation. 

\end{abstract}
\thispagestyle{empty}
\end{titlepage}

\else

\maketitle
\begin{abstract}
We study the column subset selection problem with respect to the entrywise $\ell_1$-norm loss. It is known that in the worst case, to obtain a good rank-$k$ approximation to a matrix, one needs an arbitrarily large $n^{\Omega(1)}$ number of columns to obtain a $(1+\epsilon)$-approximation to the best entrywise $\ell_1$-norm low rank approximation of an $n \times n$ matrix. Nevertheless, we show that under certain minimal and realistic distributional settings, it is possible to obtain a $(1+\epsilon)$-approximation with a nearly linear running time and poly$(k/\epsilon)+O(k\log n)$ columns. Namely, we show that if the input matrix $A$ has the form $A = B + E$, where $B$ is an arbitrary rank-$k$ matrix, and $E$ is a matrix with i.i.d. entries drawn from any distribution $\mu$ for which the $(1+\gamma)$-th moment exists, for an arbitrarily small constant $\gamma > 0$, then it is possible to obtain a $(1+\epsilon)$-approximate column subset selection to the entrywise $\ell_1$-norm in nearly linear time. Conversely we show that if the first moment does not exist, then it is not possible to obtain a $(1+\epsilon)$-approximate subset selection algorithm even if one chooses any $n^{o(1)}$ columns. This is the first algorithm of any kind for achieving a $(1+\epsilon)$-approximation for entrywise $\ell_1$-norm loss low rank approximation. 

\end{abstract}

\fi

\section{Introduction}
Numerical linear algebra algorithms are fundamental building blocks in many machine learning and data mining tasks.
A well-studied problem is low rank matrix approximation.
The most common version of the problem is also known as Principal Component Analysis (PCA), in which the goal is to find a low rank matrix to approximate a given matrix such that the Frobenius norm of the error is minimized. 
The optimal solution of this objective can be obtained via the singular value decomposition (SVD).
Hence, the problem can be solved in polynomial time.
If approximate solutions are allowed, then the running time can be made almost linear in the number of non-zero entries of the given matrix \cite{s06,cw13,mm13,nn13,bdn15,c16}.

An important variant of the PCA problem is the entrywise $\ell_1$-norm low rank matrix approximation problem. 
In this problem, instead of minimizing the Frobenius norm of the error, we seek to minimize the $\ell_1$-norm of the error. 
In particular, given an $n\times n$ input matrix $A$, and a rank parameter $k$, we want to find a matrix $B$ with rank at most $k$ such that $\|A-B\|_1$ is minimized, where for a matrix $C$, $\|C\|_1$ is defined to be $\sum_{i,j} |C_{i,j}|$.
There are several reasons for using the $\ell_1$-norm as the error measure.
For example, solutions with respect to the $\ell_1$-norm loss are usually more robust than solutions with Frobenius norm loss \cite{h64,clmw11}.
Further, the $\ell_1$-norm loss is often used as a relaxation of the $\ell_0$-loss, which has wide applications including sparse recovery, matrix completion, and robust PCA; see e.g., \cite{xcs10,clmw11}.
Although a number of algorithms have been proposed for the $\ell_1$-norm loss~\cite{kk03,kk05,kimefficient,kwak08,zlsyO12,bj12,bd13,bdb13,meng2013cyclic,mkp13,mkp14,mkcp16,park2016iteratively}, the problem is known to be NP-hard~\cite{gv15}.
The first $\ell_1$-low rank approximation with provable guarantees was proposed by \cite{swz17}.
To cope with NP-hardness, the authors gave a solution with a $\poly(k\log n)$-approximation ratio, i.e., their algorithm outputs a rank-$k$ matrix $B'\in\mathbb{R}^{n\times n}$ for which 
\begin{align}
\|A-B'\|_1\leq \alpha\cdot\min_{\rank-k\ B}\|A-B\|_1
\end{align}
for $\alpha=\poly(k\log n)$.
  The approximation ratio $\alpha$ was further improved to $O(k\log k)$ by allowing $B'$ to have a slightly larger $k'=O(k\log n)$ rank \cite{cgklrw17}. Such $B'$ with larger rank is referred to as a bicriteria solution.
  However, in high precision applications, such approximation factors are too large.
 A natural question is if one can compute a $(1+\varepsilon)$-approximate solution efficiently for $\ell_1$-norm low rank approximation.
 In fact, a $(1+\varepsilon)$-approximation algorithm was given in \cite{bbbklw17}, but the running time of their algorithm is a prohibitive $n^{\poly(k/\varepsilon)}$.
 Unfortunately, \cite{bbbklw17} shows in the worst case that a $2^{k^{\Omega(1)}}$ running time is necessary for any constant approximation given a standard conjecture in complexity theory.

{\bf Notation.} To describe our results, let us first introduce some notation. 
We will use $[n]$ to denote the set $\{1,2,\cdots,n\}$. 
We use $A_i$ to denote the $i^{\text{th}}$ column of $A$.
We use $A^j$ to denote the $j^{\text{th}}$ row of $A$.
Let $Q\subseteq[n]$. 
We use $A_Q$ to denote the matrix which is comprised of the columns of $A$ with column indices in $Q$.
Similarly, we use $A^Q$ to denote the matrix which is comprised of the rows of $A$ with row indices in $Q$. 
We use $[n]\choose t$ to denote the set of all the size-$t$ subsets of $[n]$.
Let $\| A\|_F$ denote the Frobenius norm of a matrix $A$, i.e., $\|A\|_F$ is the square root of the sum of squares of all the entries in $A$. For $1\leq p<2$, we use $\| A \|_p$ to denote the entry-wise $\ell_p$-norm of a matrix $A$, i.e., $\|A\|_p$ is the $p$-th root of the sum of $p$-th powers of the absolute values of the entries of $A$. $\|A\|_1$ is an important special case of $\|A\|_p$, which corresponds to the sum of absolute values of the entries in $A$. 
A random variable $X$ has the Cauchy distribution if its probability density function is $f(z) = \frac{1}{\pi (1+z^2)}$. 

\subsection{Our Results}
We propose an efficient bicriteria $(1+\epsilon)$-approximate column subset selection algorithm for the $\ell_1$-norm. 
We bypass the running time lower bound mentioned above by making a mild assumption on the input data, and also show that our assumption is necessary in a certain sense.

Our main algorithmic result is described as follows. 
\begin{theorem}[Informal version of Theorem~\ref{thm:dis_l1_algorithm}]\label{thm:intro_l1_algorithm}
Suppose we are given a matrix $A= A^* +\Delta \in \R^{n\times n}$, where $\rank(A^*)=k$ for $k=n^{o(1)}$, and $\Delta$ is a random matrix for which the $\Delta_{i,j}$ are i.i.d. symmetric random variables with $\E[|\Delta_{i,j}|^p]=O(\E[|\Delta_{i,j}|]^p)$ for some constant $p>1$. Let $\epsilon\in (0,1/2)$ satisfy $1/\epsilon=n^{o(1)}.$ There is an $\wt{O}(n^2+n\poly(k/\varepsilon))$\footnote{We use the notation $\wt{O}(f):= O(f\cdot \log^{O(1)} f)$.} time algorithm (Algorithm~\ref{alg:dis_l1_algorithm}) which can output a subset $S\subseteq [n]$ with $|S|\leq \poly(k/\epsilon)+O(k\log n)$ for which
 \begin{align*}
\min_{X\in \R^{|S| \times n}} \| A_S X - A \|_1 \leq (1 + \epsilon ) \| \Delta \|_1, 
 \end{align*}
holds with probability at least $99/100$.
\end{theorem}
Note the running time in Theorem \ref{thm:intro_l1_algorithm} is nearly linear in the number of non-zero entries of $A$, 
since for an $n \times n$ matrix with i.i.d. noise drawn from any continuous distribution, the number of non-zero entries of $A$ will be $n^2$ with probability $1$. 
We also show the moment assumption of Theorem \ref{thm:intro_l1_algorithm} is
necessary in the following precise sense.
\begin{theorem}[Hardness, informal version of Theorem~\ref{thm:dis_l1_hardness}]\label{thm:intro_l1_hardness}
Let $n>0$ be sufficiently large. Let $A=\eta\cdot \mathbf{1}\cdot\mathbf{1}^\top+\Delta\in\mathbb{R}^{n\times n}$ be a random matrix where $\eta=n^{c_0}$ for some sufficiently large constant $c_0,$ $\mathbf{1}\in\mathbb{R}^n$ is the all-ones vector, and $\forall i,j\in[n],\Delta_{i,j}\sim C(0,1)$ are i.i.d. standard Cauchy random variables. Let $r=n^{o(1)}.$ Then with probability at least $1-O(1/\log\log n),$ $\forall S\subseteq [n]$ with $|S|=r,$
\begin{align*}
\min_{X\in\mathbb{R}^{r\times n}} \|A_SX-A\|_1\geq 1.002\|\Delta\|_1.
\end{align*}
\end{theorem}

\subsection{Our Techniques}

For an overview of our hardness result, we refer readers to the supplementary material, namely, Appendix~\ref{sec:hard}. 
In the following, we will outline the main techniques used in our algorithm.
\paragraph{$(1+\epsilon)$-Approximate $\ell_1$-Low Rank Approximation.}

We make the following distributional assumption on the input matrix $A\in\mathbb{R}^{n\times n}$: namely, $A = A^*+ \Delta$ where $A^*$ is an arbitrary
rank-$k$ matrix and the entries of $\Delta$ are i.i.d. from any symmetric
distribution with $\E[|\Delta_{i,j}|] = 1$
and $\E[|\Delta_{i,j}|^p] = O(1)$ for any real number $p$ strictly greater than $1$, e.g., $p = 1.000001$ would
suffice. Note that such an assumption is mild compared to typical noise models which require the noise
be Gaussian or have bounded variance; in our case the random variables may even be heavy-tailed with infinite variance.
In this setting we show it is possible to obtain a subset of $\poly(k(\epsilon^{-1} + \log n))$ columns spanning a
$(1+\epsilon)$-approximation. This provably overcomes the column subset selection lower bound of \cite{swz17}
which shows for entrywise $\ell_1$-low rank approximation that there are matrices for which any subset of $\poly(k)$
columns spans at best a $k^{\Omega(1)}$-approximation.

Consider the following algorithm: sample $\poly(k/\epsilon)$  columns of $A$, and try to cover as many of the remaining columns as possible. Here, by covering a column $i$, we mean that if $A_I$ is the subset of columns sampled, then $\min_y \|A_I y - A_i\|_1 \leq (1+O(\epsilon))n$. The reason for this notion of covering is that we are able to show in Lemma \ref{lem:lower_bound_on_cost} that in this noise model,  $\|\Delta\|_1 \geq (1-\epsilon)n^2$ w.h.p., and so if we could cover every column $i$, our overall cost would be $(1+O(\epsilon))n^2$, which would give a $(1+O(\epsilon))$-approximation to the overall cost.

We will not be able to cover all columns, unfortunately, with our initial sample of $\poly(k/\epsilon)$ columns of $A$. Instead, though, we will show that we will be able to cover all but a set $T$ of $\epsilon n/(k \log k)$ of the columns. Fortunately, we show in Lemma \ref{lem:remaining_is_small} another property of the noise matrix $\Delta$ is that {\it all} subsets $S$ of columns of size at most $n/r$, for $r \geq (1/\gamma)^{1+1/(p-1)}$ satisfy $\sum_{j \in S}\|\Delta_j\|_1 = O(\gamma n^2)$. Thus, for the above set $T$ that we do not cover, we can apply this lemma to it with $\gamma = \epsilon/(k \log k)$, and then we know that $\sum_{j \in T}\|\Delta_j\|_1 = O(\epsilon n^2/(k \log k))$, which then enables us to run a previous $\tilde{O}(k)$-approximate $\ell_1$ low rank approximation algorithm \cite{cgklrw17} on the set $T$, which will only incur total cost $O(\epsilon n^2)$, and since by Lemma \ref{lem:lower_bound_on_cost} above the overall cost is at least $(1-\epsilon)n^2$, we can still obtain a $(1+O(\epsilon))$-approximation overall.

The main missing piece of the algorithm to describe is why we are able to cover all but a small fraction of the columns. One thing to note is that our noise distribution may not have a finite variance, and consequently, there {\it can} be very large entries $\Delta_{i,j}$ in some columns. In Lemma \ref{lem:hard_one_is_small}, we show the number of columns in $\Delta$ for which there exists an entry larger than $n^{1/2+1/(2p)}$ in magnitude is $O(n^{(2-p)/2})$, which since $p > 1$ is a constant bounded away from $1$, is sublinear. Let us call this set with entries larger than $n^{1/2+1/(2p)}$ in magnitude the set $H$ of ``heavy" columns; we will not make any guarantees about $H$, rather, we will stuff it into the small set $T$ of columns above on which we will run our earlier $O(k \log k)$-approximation.

For the remaining, non-heavy columns, which constitute almost all of our columns, we show in Lemma \ref{lem:easy_one_is_concentrated} that $\|\Delta_i\|_1 \leq (1+\epsilon) n$ w.h.p. The reason this is important is that recall to cover some column $i$ by a sample set $I$ of columns, we need $\min_y \|A_I y - A_i\|_1 \leq (1+O(\epsilon))n$. It turns out, as we now explain, that we will get $\min_y \|A_Iy-A_i\|_1 \leq \|\Delta_i\|_1 + e_i$, where $e_i$ is a quantity which we can control and make $O(\epsilon n)$ by increasing our sample size $I$. Consequently, since $\|\Delta_i\|_1 \leq (1+\epsilon) n$, overall we will have $\min_y \|A_Iy-A_i\|_1 \leq (1+O(\epsilon))n$, which means that $i$ will be covered. We now explain what $e_i$ is, and why $\min_y \|A_Iy-A_i\|_1 \leq \|\Delta_i\|_1 + e_i$.

Towards this end, we first explain a key insight in this model. Since the $p$-th moment exists for some real number $p > 1$ (e.g., $p = 1.000001$ suffices), {\it averaging} helps reduce the noise of fitting a column $A_i$ by subsets of other columns.  Namely, we show in Lemma \ref{lem:averaging_works} that for any $t$ non-heavy column $\Delta_{i_1}, \ldots, \Delta_{i_t}$ of $\Delta$, and any coefficients $\alpha_1, \alpha_2, \ldots, \alpha_t \in [-1,1]$, $\|\sum_{j=1}^t \alpha_j \Delta_{i_j}\|_1 = O(t^{1/p} n)$, that is, since the individual coordinates of the $\Delta_{i_j}$ are zero-mean random variables, their sum {\it concentrates} as we add up more columns. We do not need bounded variance for this property.

How can we use this averaging property for subset selection? The idea is, instead of sampling a single subset $I$ of $O(k)$ columns and trying to cover each remaining column with this subset as shown in \cite{cgklrw17}, we will sample multiple independent subsets $I_1, I_2, \ldots, I_t$. Each set has size $\poly(k/\varepsilon)$ and we will sample at most $\poly(k/\varepsilon)$ subsets. 
By a similar argument of \cite{cgklrw17},
 for any given column index $i \in [n]$, for most of these subset $I_j$, we have that $A^*_i/\|\Delta_i\|_1$ can be expressed as a linear combination of columns $A^*_{\ell}/\|\Delta_{\ell}\|_1, \ell \in I_j,$ via coefficients of absolute value at most $1$. Note that this is only true for most $i$ and most $j$; we develop terminology for this in Definitions \ref{def:tuple}, \ref{def:core_and_good_tuple}, \ref{def:core_tuple}, and \ref{def:coefficients_tuple}, referring to what we call a {\it good core}. We quantify what we mean by most $i$ and most $j$ having this property in Lemma \ref{lem:label_uniform_samples} and Lemma \ref{lem:easy_to_find_good_tuple}.

The key though, that drives the analysis, is Lemma \ref{lem:good_tuple_low_cost}, which shows that $\min_y \|A_iy-A_i\|_1 \leq \|\Delta_i\|_1 + e_i$, where $e_i = O(q^{1/p}/t^{1-1/p} n)$, where $q$ is the size of each $I_j$, and $t$ is the number of different $I_j$. We need $q$ to be at least $k$, just as before, so that we can be guaranteed that when we adjoin a column index $i$ to $I_j$,
there is some positive probability that $A^*_i/\|\Delta_i\|_1$ can be expressed as a linear combination of columns $A^*_{\ell}/\|\Delta_{\ell}\|_1, \ell \in I_j$, with coefficients of absolute value at most $1$. What is different in our noise model though is the division by $t^{1-1/p}$. Since $p > 1$, if we set $t$ to be a large enough $\poly(k/\epsilon)$, then $e_i = O(\epsilon n)$, and then we will have covered $A_i$, as desired. This captures the main property that averaging the linear combinations for expression $A^*_i/\|\Delta_i\|_1$ using different subsets $I_j$ gives us better and better approximations to $A^*_i/\|\Delta_i\|_1$. Of course we need to ensure several properties such as not sampling a heavy column (the averaging in Lemma \ref{lem:averaging_works} does not apply when this happens), we need to ensure most of the $I_j$ have small-coefficient linear combinations expressing $A^*_i/\|\Delta_i\|_1$, etc. This is handled in our main theorem, Theorem \ref{thm:dis_l1_algorithm}.

\section{$\ell_1$-Norm Column Subset Selection}\label{sec:dis}

We first present two subroutines.

{\bf Linear regression with $\ell_1$ loss.} The first subroutine needed is an approximate $\ell_1$ linear regression solver.
In particular, given a matrix $M\in\mathbb{R}^{n\times d}$, $n$ vectors $b_1,b_2,\cdots,b_n\in\mathbb{R}^n$, and an error parameter $\varepsilon\in(0,1)$, we want to compute $x_1,x_2,\cdots,x_n\in\mathbb{R}^d$ for which $\forall i\in[n]$, we have 
\begin{align*}
\|Mx_i-b_i\|_1\leq (1+\varepsilon)\cdot\min_{x\in\mathbb{R}^d} \|Mx-b_i\|_1.
\end{align*}
Furthermore, we also need an estimate $v_i$ of the regression cost $\|Mx_i-b_i\|_1$ for each $i\in[n]$ such that $\|Mx_i-b_i\|_1\leq v_i\leq (1+\varepsilon)\|Mx_i-b_i\|_1$.
Such an $\ell_1$-regression problem can be solved efficiently (see \cite{w14} for a survey). 
The total running time to solve these $n$ regression problems simultaneously is at most $\wt{O}(n^2)+n\cdot \poly(d\log n)$, and the success probability is at least $0.999$.

{\bf $\ell_1$ Column subset selection for general matrices.} 
The second subroutine needed is an $\ell_1$-low rank approximation solver for general input matrices, though we allow a large approximation ratio.
We use the algorithm proposed by \cite{cgklrw17} for this purpose. 
In particular, given an $n\times d$ $(d\leq n)$ matrix $M$ and a rank parameter $k$, the algorithm can output a small set $S\subset[n]$ with size at most $O(k\log n)$, such that 
\begin{align*}
\min_{X\in\mathbb{R}^{|S|\times d}}\|M_SX-M\|_1\leq O(k\log k)\cdot \min_{\rank-k\ B}\|M-B\|_1.
\end{align*}
Furthermore, the running time is at most $\wt{O}(n^2)+n\cdot\poly(k\log n)$, and the success probability is at least $0.999$. Now we can present our algorithm, Algorithm~\ref{alg:dis_l1_algorithm}.
\begin{algorithm}
	\begin{algorithmic}[1]\caption{$\ell_1$-Low Rank Approximation with Input Assumption}\label{alg:dis_l1_algorithm}
		\Procedure{\textsc{L1NoisyLowRankApprox}}{$A\in\mathbb{R}^{n\times n},k,\varepsilon$} \Comment{Theorem~\ref{thm:dis_l1_algorithm}}
		\State Sample a set $I$ from ${[n] \choose s}$ uniformly at random, where $s=\poly(k/\varepsilon).$ \label{sta:uniform_sample}
		\State Solve the approximate $\ell_1$-regression problem $\min_{x\in\mathbb{R}^{|I|}}\|A_I x - A_i\|_1$ for each $i\in[n]$, and let $v_i$ be the estimated regression cost. 
		\label{sta:l1_regression}
		\State Compute the set $T=\{i\in[n]\mid v_i\text{ is one of the top }l\text{ largest values among } v_1,v_2,\cdots,v_n\}$, where $l=n/\poly(k/\varepsilon)$. 
		\State Solve $\ell_1$-column subset selection for $A_T.$ Let the solution be $A_Q$. 
		\State Solve the approximate $\ell_1$-regression problem $\min_{X\in\mathbb{R}^{(|I|+|Q|)\times n}}\|A_{(I\cup Q)}X - A\|_1$, and let $\hat{X}$ be the solution.
		Return $A_{(I\cup Q)}$ and $\hat{X}$. \Comment{$A_{(I\cup Q)}\hat{X}$ is a good low rank approximation to $A$}
		\EndProcedure
	\end{algorithmic}
\end{algorithm}

{\bf Running time.} Uniformly sampling a set $I$ can be done in $\poly(k/\varepsilon)$ time.
According to our $\ell_1$-regression subroutine, solving $\min_x \|A_Ix-A_i\|_1$ for all $i\in[n]$ can be finished in $\wt{O}(n^2)+n\cdot\poly(k\log(n)/\varepsilon)$ time.
We only need sorting to compute the set $T$ which takes $O(n\log n)$ time.
By our second subroutine, the $\ell_1$-column subset selection for $A_T$ will take $\wt{O}(n^2)+n\cdot\poly(k\log n)$.
The last step only needs an $\ell_1$-regression solver, which takes $\wt{O}(n^2)+n\cdot\poly(k\log(n)/\varepsilon)$ time.
Thus, the overall running time is $\wt{O}(n^2)+n\cdot\poly(k\log(n)/\varepsilon)$.

The remaining parts in this section will focus on analyzing the correctness of the algorithm.

\subsection{
Properties of the Noise Matrix
}

Recall that the input matrix $A\in\mathbb{R}^{n\times n}$ can be decomposed as $A^*+\Delta$, where $A^*$ is the ground truth, and $\Delta$ is a random noise matrix. 
In particular, $A^*$ is an arbitrary rank-$k$ matrix, and $\Delta$ is a random matrix where each entry is an i.i.d. sample drawn from an unknown symmetric distribution. 
The only assumption on $\Delta$ is that each entry $\Delta_{i,j}$ satisfies $\E[|\Delta_{i,j}|^p]=O(\E[|\Delta_{i,j}|^p])$ for some constant $p>1$, i.e., the $p$-th moment of the noise distribution is bounded.
Without loss of generality, we will suppose $\E[|\Delta_{i,j}|]=1$, $\E[|\Delta_{i,j}|^p]=O(1)$, and $p\in(1,2)$ throughout the paper. In this section, we will present some key properties of the noise matrix.

The following lemma provides a lower bound on $\|\Delta\|_1$. 
Once we have the such lower bound, we can focus on finding a solution for which the approximation cost is at most that lower bound.
\begin{lemma}[Lower bound on the noise matrix]\label{lem:lower_bound_on_cost}
Let $\Delta\in\mathbb{R}^{n\times n}$ be a random matrix where $\Delta_{i,j}$ are i.i.d. samples drawn from a symmetric distribution. 
Suppose $\E[|\Delta_{i,j}|]=1$ and $\E[|\Delta_{i,j}|^p]=O(1)$ for some constant $p\in(1,2)$. 
Then, $\forall \epsilon\in(0,1)$ which satisfies $1/\epsilon=n^{o(1)},$ we have
\begin{align*}
\Pr \left[ \|\Delta\|_1\geq(1-\epsilon)n^2 \right]\geq 1-e^{-\Theta(n)}.
\end{align*}
\end{lemma}

The next lemma shows the main reason why we are able to get a small fitting cost when running regression.
Consider a toy example.
Suppose we have a target number $a\in\mathbb{R}$, and another $t$ numbers $a+g_1,a+g_2,\cdots,a+g_t\in\mathbb{R}$, where $g_i$ are i.i.d. samples drawn from the standard Gaussian distribution $N(0,1)$. 
If we use $a+g_i$ to fit $a$, then the expected cost is $\E[|a+g_i-a|]=\E[|g_i|]=\sqrt{2/\pi}$.
However, if we use the average of $a+g_1,a+g_2,\cdots,a+g_t$ to fit $a$, then the expected cost is $\E[|\sum_{i=1}^t g_i|/t]$. 
Since the $g_i$ are independent, $\sum_{i=1}^t g_i$ is a random Gaussian variable with variance $t$, which means that the above expected cost is $\sqrt{2/\pi }/\sqrt{t}$.
Thus the fitting cost is reduced by a factor $\sqrt{t}$. 
By generalizing the above argument, we obtain the following lemma. 

\begin{lemma}[Averaging reduces the noise]\label{lem:averaging_works}
Let $\Delta_1,\Delta_2,\cdots,\Delta_t\in\mathbb{R}^n$ be $t$ random vectors. The $\Delta_{i,j}$ are i.i.d. symmetric random variables with $\E[|\Delta_{i,j}|]=1$ and $\E[|\Delta_{i,j}|^p]=O(1)$ for some constant $p\in(1,2)$. Let $\alpha_1,\alpha_2,\cdots,\alpha_t\in[-1,1]$ be $t$ real numbers. Conditioned on $\forall i\in[n],j\in[t],|\Delta_{i,j}|\leq n^{1/2+1/(2p)},$ with probability at least $1-2^{-n^{\Theta(1)}},$
\begin{align*}
\left\|\sum_{i=1}^t\alpha_i\Delta_i \right\|_1\leq O(t^{1/p}n).
\end{align*}
\end{lemma}

The above lemma needs a condition that each entry in the noise column should not be too large.
Fortunately, we can show that most of the (noise) columns do not have any large entry.
\begin{lemma}[Only a small number of columns have large entries]\label{lem:hard_one_is_small}
	Let $\Delta\in\mathbb{R}^{n\times n}$ be a random matrix where the $\Delta_{i,j}$ are i.i.d. symmetric random variables with $\E[|\Delta_{i,j}|]=1$ and $\E[|\Delta_{i,j}|^p]=O(1)$ for some constant $p\in(1,2)$. Let
	\begin{align*}
	H=\{j\in[n] ~\big|~ \exists i\in[n],|\Delta_{i,j}|>n^{1/2+1/(2p)}\}.
	\end{align*}
	Then with probability at least $0.999$ $|H|\leq O(n^{1-(p-1)/2}).$
\end{lemma}

The following lemma shows that any small subset of the columns of the noise matrix $\Delta$ cannot contribute too much to the overall error.
By combining with the previous lemma, the entrywise $\ell_1$ cost of all columns containing large entries can be bounded.
\begin{lemma}
	\label{lem:remaining_is_small}
Let $\Delta\in\mathbb{R}^{n\times n}$ be a random matrix where $\Delta_{i,j}$ are i.i.d. symmetric random variables with $\E[|\Delta_{i,j}|]=1$ and $\E[|\Delta_{i,j}|^p]=O(1)$ for some constant $p\in(1,2)$. Let $\epsilon\in(0,1)$ satisfy $1/\epsilon=n^{o(1)}.$ Let $r\geq (1/\epsilon)^{1+1/(p-1)}.$ Then, with probability at least $.999,$ $\forall S\subset [n]$ with $|S|\leq n/r,$ $\sum_{j\in S}\|\Delta_j\|_1 =  O(\epsilon n^2).$
\end{lemma}

We say a (noise) column is good if it does not have a large entry.
We can show that, with high probability, the entry-wise $\ell_1$ cost of a good (noise) column is small.

\begin{lemma}[Cost of good noise columns]\label{lem:easy_one_is_concentrated}
Let $\Delta\in\mathbb{R}^{n}$ be a random vector where $\Delta_{i}$ are i.i.d. symmetric random variables with $\E[|\Delta_{i}|]=1$ and $\E[|\Delta_{i}|^p]=O(1)$ for some constant $p\in(1,2)$. Let $\epsilon\in(0,1)$ satisfy $1/\epsilon=n^{o(1)}.$ If $\forall i\in[n],|\Delta_i|\leq n^{1/2+1/(2p)},$ then with probability at least $1-2^{-n^{\Theta(1)}},$ $\|\Delta\|_1\leq (1+\epsilon)n.$
\end{lemma}

\subsection{Definition of Tuples and Cores}
In this section, we provide some basic definitions, e.g., of a tuple, a good tuple, the core of a tuple, and a coefficients tuple. These definitions will be heavily used later when we analyze the correctness of our algorithm.

Before we present the definitions, we introduce a notion $R_{A^*}(S)$.
Given a matrix $A^* \in \R^{n_1 \times n_2}$,
for a set $S \subseteq [n_2]$, we define
\begin{align*}
R_{A^*}(S) := \arg\max_{P:P \subseteq S} \left\{ \left|\det\left( (A^*)_P^Q \right)\right| ~\bigg|~ |P| = |Q|= \rank (A^*_S), Q \subseteq [n_1]  \right\},
\end{align*}
where for a squared matrix $C$, $\det(C)$ denotes the determinant of $C$.
The above maximum is over both $P$ and $Q$ while $R_{A^*}(S)$ only takes the value of the corresponding $P$.

By Cramer's rule, if we use the columns of $A^*$ with index in the set $R_{A^*}(S)$ to fit any column of $A^*$ with index in the set $S$, the absolute value of any fitting coefficient will be at most $1$. 
The use of Cramer's rule is as follows.
Consider a rank $k$ matrix $M\in\mathbb{R}^{n\times(k+1)}$.
Let $P\subseteq[k+1],Q\subseteq [n],|P|=|Q|=k$ be such that $|\det(M_P^Q)|$ is maximized.
Since $M$ has rank $k$, we know $\det(M_P^Q)\not= 0$ and thus the columns of $M_P$ are independent.
Let $i\in [k+1]\setminus P$.
Then the linear equation $M_Px=M_i$ is feasible and
there is a unique solution $x$.
Furthermore, by Cramer's rule $x_j={\det(M^Q_{[k+1]\setminus\{j\}})}/{\det(M_P^Q)}$. 
Since $|\det(M_P^Q)|\geq |\det(M^Q_{[k+1]\setminus\{j\}})|$, we have $\|x\|_{\infty}\leq 1$.

Small fitting coefficients are good since they will not increase the noise by too much.
For example, suppose $A^*_i=A^*_S x$ and $\|x\|_{\infty}\leq 1$, i.e., the $i$-th column can be fit by the columns with indices in the set $S$ and the fitting coefficients $x\in\mathbb{R}^{|S|}$ are small. 
If we use the noisy columns of $A^*_S+\Delta_S$ to fit the noisy column $A^*_i+\Delta_i$, then the fitting cost is at most $\|(A^*_S+\Delta_S)x-(A^*_i+\Delta_i)\|_1\leq \|\Delta_i\|_1+\|\Delta_S x\|_1$.
Since $\|x\|_{\infty}\leq 1$, it is possible to give a good upper bound for $\|\Delta_S x\|_1$. 

\begin{definition}[Tuple]\label{def:tuple}
A $(q,t,n)-$tuple is defined to be
$
(S_1,S_2,\cdots,S_t,i),
$
where $\forall j\in[t],S_j\subset [n]$ with $|S_j|=q.$ Let $S=\bigcup_{j=1}^t S_j$. 
Then $|S|=qt,$ i.e., $S_1,S_2,\cdots,S_t$ are disjoint. Furthermore, $i\in[n]$ and $i\not\in S.$ For simplicity, we use $(S_{[t]},i)$ to denote $(S_1,S_2,\cdots,S_t,i)$.
\end{definition}

We next provide the definition of a good tuple.
\begin{definition}[Good tuple]\label{def:core_and_good_tuple}
Given a $\rank$-$k$ matrix $A^*\in\mathbb{R}^{n\times n},$ an $(A^*,q,t,\alpha)$-good tuple is a $(q,t,n)$-tuple $
(S_{[t]},i)
$ which satisfies
\begin{align*}
|\{j\in[t]\mid i\not\in R_{A^*}(S_j\cup \{i\})\}|\geq \alpha \cdot t.
\end{align*}
\end{definition}

We need the definition of the core of a tuple.
\begin{definition}[Core of a tuple]\label{def:core_tuple}
The core of $(S_{[t]},i)$ is defined to be the set
\begin{align*}
\{j\in[t]\mid i\not\in R_{A^*}(S_j\cup \{i\})\}.
\end{align*}
\end{definition}

We define a coefficients tuple as follows.
\begin{definition}[Coefficients tuple]\label{def:coefficients_tuple}
Given a $\rank$-$k$ matrix $A^*\in\mathbb{R}^{n\times n},$ let $(S_{[t]},i)$ be an $(A^*,q,t,\alpha)$-good tuple. Let $C$ be the core of $(S_{[t]},i)$. A coefficients tuple corresponding to $(S_{[t]},i)$ is defined to be $(x_1,x_2,\cdots,x_t)$ where $\forall j\in[t],x_j \in \R^q.$  The vector $x_j \in \R^q$ satisfies:
$x_j = 0$ if $j \in [t] \backslash C$, while 
$A_{S_j}^* x_j = A_i^*$ and $\| x_j \|_{\infty}\leq 1$, if $j \in C$. 
To guarantee the coefficients tuple is unique, we restrict each vector $x_j \in \R^q$ to be one that has the minimum lexicographic order.
\end{definition}

\subsection{Properties of a Good Tuple and a Coefficients Tuple}

Consider a good tuple $(S_1,S_2,\cdots,S_t,i)$. 
By the definition of a good tuple, the size of the core $C$ of the tuple is large.
For each $j\in C$, the coefficients $x_j$ of using $A^*_{S_j}$ to fit $A^*_i$ should have absolute value at most $1$.
Now consider the noisy setting.
As discussed in the previous section, using $A_{S_j}$ to fit $A_i$ has cost at most $\|\Delta_i\|_1+\|\Delta_{S_j}x_j\|_1$.
Although $\|\Delta_{S_j}x_j\|_1$ has a good upper bound, it is not small enough.
To further reduce the $\ell_1$ fitting cost, we can now apply the averaging argument (Lemma~\ref{lem:averaging_works}) over all the fitting choices corresponding to $C$.
Formally, we have the following lemma.

\begin{lemma}[Good tuples imply low fitting cost]\label{lem:good_tuple_low_cost}
Suppose we are given a matrix $A\in\mathbb{R}^{n\times n}$ which satisfies $A=A^*+\Delta$, where $A^*\in\mathbb{R}^{n\times n}$ has rank $k$. Here $\Delta\in\mathbb{R}^{n\times n}$ is a random matrix where $\Delta_{i,j}$ are i.i.d. symmetric random variables with $\E[|\Delta_{i,j}|]=1$ and $\E[|\Delta_{i,j}|^p]=O(1)$ for some constant $p\in(1,2)$. Let $H \subset [n]$ be defined as follows:
\begin{align*}
 H= \left\{ j\in[n] ~\bigg|~ \exists i\in[n],|\Delta_{i,j}|>n^{1/2+1/(2p)} \right\}.
\end{align*}
Let $q,t\leq n^{o(1)}.$ Then, with probability at least $1-2^{-n^{\Theta(1)}},$ for all $(A^*,q,t,1/2)$-good tuples $(S_1,S_2,\cdots,S_t,i)$ which satisfy $H\cap\left(\bigcup_{j=1}^t S_j\right)=\emptyset,$ we have 
\begin{align*}
\min_{y\in\mathbb{R}^{qt}}\left\|A_{\{\bigcup_{j=1}^t S_j\}}y-A_i\right\|_1\leq\left\|\frac{1}{|C|}\sum_{j=1}^t A_{S_j}x_j-A_i\right\|_1\leq \|\Delta_i\|_1+O(q^{1/p}/t^{1-1/p}n),
\end{align*}
where $C$ is the core of $(S_1,S_2,\cdots,S_t,i)$, and $(x_1,x_2,\cdots,x_t)$ is the coefficients tuple corresponding to $(S_1,S_2,\cdots,S_t,i).$
\end{lemma}

We next show that if we choose columns randomly, it is easy to find a good tuple.

\begin{lemma}
	\label{lem:label_uniform_samples}
	Given a rank-$k$ matrix $A^*\in\mathbb{R}^{n\times n}$, let $q>10k,t>0.$ Let $I=\{i_1,i_2,\cdots,i_{qt+1}\}$ be a subset drawn uniformly at random from $[n]\choose qt+1$. Let $\pi : I \rightarrow I$ be a random permutation of $qt+1$ elements. $\forall j\in[t],$ let
	\begin{align*}
	S_j= \left\{ i_{\pi((j-1)q+1)},i_{\pi((j-1)q+2)},\cdots,i_{\pi((j-1)q+q)} \right\}. 
	\end{align*}
	We use $i$ to denote $i_{\pi(qt+1)}$. With probability $\geq 1-2k/q$, $(S_1,S_2,\cdots,S_t,i)$ is an $(A^*,q,t,1/2)-$good tuple.
\end{lemma}

Lemma~\ref{lem:label_uniform_samples} implies that if we randomly choose $S_1,S_2,\cdots,S_t$, then with high probability, there are many choices of $i\in[n]$, such that $(S_1,S_2,\cdots,S_t,i)$ is a good tuple.
Precisely, we can show the following.

\begin{lemma}
	\label{lem:easy_to_find_good_tuple}
Given a rank-$k$ matrix $A^*\in\mathbb{R}^{n\times n}$, let $q>10k,t>0.$ Let $I=\{i_1,i_2,\cdots,i_{qt}\}$ be a random subset uniformly drawn from $[n]\choose qt$. Let $\pi$ be a random permutation of $qt$ elements. $\forall j\in[t],$ we define $S_j$ as follows:
\begin{align*}
S_j= \left\{ i_{\pi((j-1)q+1)}, i_{\pi((j-1)q+2)}, \cdots, i_{\pi((j-1)q+q)} \right\}.
\end{align*}
Then with probability at least $2k/q$,
\begin{align*}
\left|\left\{i\in[n]\setminus I ~\big|~ (S_1,S_2,\cdots,S_t,i)\mathrm{~is~an~}(A^*,q,t,1/2)\mathrm{-good~tuple~}\right\}\right|\geq (1-4k/q)(n-qt).
\end{align*}
\end{lemma}

\subsection{Main Result}
Now we are able to put all ingredients together to prove our main theorem,  
Theorem~\ref{thm:dis_l1_algorithm}.

\begin{theorem}[Formal version of Theorem~\ref{thm:intro_l1_algorithm}]\label{thm:dis_l1_algorithm}
Suppose we are given a matrix $A= A^* +\Delta \in \R^{n\times n}$, where $\rank(A^*)=k$ for $k=n^{o(1)}$, and $\Delta$ is a random matrix for which the $\Delta_{i,j}$ are i.i.d. symmetric random variables with $\E[|\Delta_{i,j}|]=1$ and $\E[|\Delta_{i,j}|^p]=O(1)$ for some constant $p\in(1,2)$. Let $\epsilon\in (0,1/2)$ satisfy $1/\epsilon=n^{o(1)}.$ There is an $\wt{O}(n^2+n\poly(k/\varepsilon))$ time algorithm (Algorithm~\ref{alg:dis_l1_algorithm}) which can output a subset $S\in [n]$ with $|S|\leq \poly(k/\epsilon)+O(k\log n)$ for which 
 \begin{align*}
\min_{X\in \R^{|S| \times n}} \| A_S X - A \|_1 \leq (1 + \epsilon ) \| \Delta \|_1,
 \end{align*}
holds with probability at least $99/100$.
\end{theorem}
\begin{proof}
We discussed the running time at the beginning of  Section~\ref{sec:dis}. Next, we turn to correctness. Let $q=\Omega\left(\frac{k(k\log k)^{1+\frac{1}{p-1}}}{\epsilon^{1+\frac{1}{p-1}}}\right),$ $t=\frac{q^{\frac{1}{p-1}}}{\epsilon^{1+\frac{1}{p-1}}}.$ Let $r=\Theta(q/k).$ Let
\begin{align*}
 I_1= \left\{ i_1^{(1)},i_2^{(1)},\cdots,i_{qt}^{(1)} \right\}, I_2= \left\{i_1^{(2)},i_2^{(2)},\cdots,i_{qt}^{(2)} \right\},\cdots,I_r= \left\{i_1^{(r)},i_2^{(r)},\cdots,i_{qt}^{(r)} \right\},
\end{align*}
 be $r$ independent subsets drawn uniformly at random from  $[n]\choose qt$. 
 Let $I=\bigcup_{s\in[r]}I_s$, which is the same as that in Algorithm~\ref{alg:dis_l1_algorithm}.
 Let $\pi_1,\pi_2,\cdots,\pi_{r}$ be $r$ independent random permutations of $qt$ elements. Due to Lemma~\ref{lem:easy_to_find_good_tuple} and a Chernoff bound, with probability at least $.999,$ $\exists s\in[r]$,  
\begin{align*}
\left|\left\{i\in[n]\setminus I_s ~\big|~ (S_1,S_2,\cdots,S_t,i)\text{~is an $(A^*,q,t,1/2)-$good tuple~}\right\}\right|\geq (1-4k/q)(n-qt)
\end{align*}
where
\begin{align*}
S_j= \left\{ i^{(s)}_{\pi_s((j-1)q+1)},i^{(s)}_{\pi_s((j-1)q+2)},\cdots,i^{(s)}_{\pi_s((j-1)q+q)} \right\}, \forall j \in [t].
\end{align*}

Let set $H \subset [n]$ be defined as follows:
\begin{align*}
H=\{j\in[n]\mid \exists i\in[n],|\Delta_{i,j}|>n^{1/2+1/(2p)}\}.
\end{align*}
 Then due to Lemma~\ref{lem:hard_one_is_small}, with probability at least $0.999,$ $|H|\leq O(n^{1-(p-1)/2}).$ Thus, for $j\in[r],$ the probability that $H\cap I_j\not=\emptyset$ is at most $O(qt\cdot n^{1-(p-1)/2}/(n-qt))=1/n^{\Omega(1)}.$ By taking a union bound over all $j\in[r],$ with probability at least $1-1/n^{\Omega(1)},$ $\forall j\in[r],I_j\cap H=\emptyset.$ Thus, we can condition on $I_s\cap H=\emptyset.$
Due to Lemma~\ref{lem:good_tuple_low_cost} and $q^{1/p}/t^{1-1/p}=\epsilon,$
\begin{align*}
\left|\left\{i\in[n]\setminus I_s ~\bigg|~ \min_{y\in\mathbb{R}^{qt}}\|A_{I_s}y-A_i\|_1\leq \|\Delta_i\|_1+O(\epsilon n)\right\}\right|\geq (1-4k/q)(n-qt). 
\end{align*}
 Due to Lemma~\ref{lem:easy_one_is_concentrated} and a union bound over all $i\in[n]\setminus H$, with probability at least $.999,$ $\forall i\not\in H,\|\Delta_i\|\leq (1+\epsilon)n$. Thus,
 \begin{align*}
\left|\left\{i\in[n]\setminus I_s ~\bigg|~ \min_{y\in\mathbb{R}^{qt}}\|A_{I_s}y-A_i\|_1\leq (1+O(\epsilon)) n\right\}\right|\geq (1-4k/q)(n-qt)-|H|.
\end{align*}
Let
 \begin{align*}
T'=[n]\setminus\left\{i\in[n] ~\bigg|~ \min_{y\in\mathbb{R}^{qt}}\|A_{I_s}y-A_i\|_1\leq (1+O(\epsilon)) n\right\}.
\end{align*}
Then $|T'|\leq O(kn/q+n^{1-(p-1)/2})=O(kn/q)=O(( \epsilon / ( k\log k ) )^{1+ 1 /(p-1)}n).$ 
By our selection of $T$ in algorithm~\ref{alg:dis_l1_algorithm}, $T'$ should be a subset of $T$.
Due to Lemma~\ref{lem:remaining_is_small}, with probability at least $.999,$ $\|\Delta_{T}\|_1\leq O(\epsilon n^2/(k\log k) )$. By our second subroutine mentioned at the beginning of Section~\ref{sec:dis}
it can find a set $Q\subset[n]$ with $|Q|=O(k\log n)$ such that
$min_{X\in\mathbb{R}^{|Q|\times |T|}}\|A_QX-A_T\|_1\leq O(k\log k)\|\Delta_T\|_1\leq O(\epsilon n^2).$
Thus, we have
$\min_{X\in\mathbb{R}^{(|Q|+q\cdot t\cdot r)\times n}}\|A_{(Q\cup I)}X-A\|_1
\leq 
\min_{X_1\in\mathbb{R}^{(q\cdot t\cdot r)\times n}}\|A_{ I}X_1-A_{[n]\setminus T}\|_1 + \min_{X_2\in\mathbb{R}^{|Q|\times n}}\|A_{ Q}X_2-A_{T}\|_1
\leq (1+O(\epsilon)) n^2.$
Due to Lemma~\ref{lem:lower_bound_on_cost}, with probability at least $.999,$ $\|\Delta\|_1\geq (1-\epsilon)n^2,$ and thus 
$\min_{X\in\mathbb{R}^{(|Q|+q\cdot t\cdot r)\times n}}\|A_{(Q\cup  I ) }X-A\|_1\leq (1+O(\epsilon)) \|\Delta\|_1.$

\end{proof}

\section{Experiments}\label{sec:exp}
The take-home message from our theoretical analysis is that although the noise distribution may be heavy-tailed, if the $p$-th $(p>1)$ moment of the distribution exists, averaging the noise may reduce the noise. 
In the spirit of averaging, we found that taking a median works a bit better in practice.
Inspired by our theoretical analysis, we propose a simple heuristic algorithm (Algorithm~\ref{alg:heu}) which can output a rank-$k$ solution. We tested Algorithm~\ref{alg:heu} on both synthetic and real datasets.

\begin{algorithm}
	\begin{algorithmic}[1]\caption{Median Heuristic}\label{alg:heu}
		\Procedure{\textsc{L1NoisyLowRankApproxHeu}}{$A\in\mathbb{R}^{n\times d},k\geq 1$} 
		\State Sample a set $I=\{i_1,i_2,\cdots,i_{sk}\}$  from ${[n] \choose sk}$ uniformly at random.
		\State Compute $B\in\mathbb{R}^{n\times k}$ s.t., for $t\in[n],q\in[k],$ $B_{t,q}=\mathrm{median}(A_{t,i_{s(q-1)+1}},\cdots,A_{t,i_{sq}})$.
		
		\State  Solve $\min_{X\in\mathbb{R}^{k\times d}}\|BX-A\|_1$ and let the solution be $X^*$. Output $BX^*$.
		\EndProcedure
	\end{algorithmic}
\end{algorithm}

{\bf Datasets.} 
For each rank-$k$ experiment, we chose a high rank matrix $\hat{A}\in\mathbb{R}^{n\times d}$, applied top-$k$ SVD to $\hat{A}$ and obtained a rank-$k$ matrix $A^*$ as our ground truth matrix.
For our synthetic data experiments, the matrix $\hat{A}\in\mathbb{R}^{500\times 500}$ was generated at random, where each entry was drawn uniformly from $\{0,1,\cdots,9\}$. For real datasets, we chose \textit{isolet}\footnote{\url{https://archive.ics.uci.edu/ml/datasets/isolet}} $(617 \times 1559)$ or \textit{mfeat}\footnote{\url{https://archive.ics.uci.edu/ml/datasets/Multiple+Features}} $(651\times 2000)$ as $\hat{A}$~\cite{an07}. We tested two different noise distributions. One distribution is the standard L\'evy $1.1$-stable distribution~\cite{m60}. Another distribution is constructed from the standard Cauchy distribution, i.e., to draw a sample from the constructed distribution, we draw a sample from the Cauchy distribution, keep the sign unchanged, and take the $\frac{1}{1.1}$-th power of the absolute value.
Notice that both distributions have bounded $1.1$-th moment, but do not have a $p$-th moment for any $p>1.1$. 
To construct the noise matrix $\Delta\in\mathbb{R}^{n\times d}$, we drew a matrix $\hat{\Delta}$ where each entry is an i.i.d. sample from one of the two noise distributions, and then scaled the noise: $\Delta = \hat{\Delta}\cdot \frac{\|A^*\|_1}{20\cdot n\cdot d}$.
 We set $A=A^*+\Delta$ as the input.

{\bf Methodologies.} We compare Algorithm~\ref{alg:heu} with SVD, $\poly(k,\log n)$-approximate entrywise $\ell_1$ low rank approximation~\cite{swz17}, and uniform $k$-column subset sampling~\cite{cgklrw17}\footnote{We chose to compare with \cite{swz17,cgklrw17} due to their theoretical guarantees. Though the uniform $k$-column subset sampling described in the experiments of \cite{cgklrw17} is a heuristic algorithm, it is inspired by their theoretical algorithm.}.
 For Algorithm~\ref{alg:heu}, we set $s=\min(50,\lfloor n/k\rfloor)$.
 For all of algorithms 
we repeated the experiment the same number of times and compared the best solution obtained by each algorithm. 
 We report the approximation ratio $\|B-A\|_1/\|\Delta\|_1$ for each algorithm, where $B\in\mathbb{R}^{n\times d}$ is the output rank-$k$ matrix. The results are shown in Figure~\ref{fig:results}. As shown in the figure, Algorithm~\ref{alg:heu} outperformed all of the other algorithms.

 \begin{figure*} 
 	
 	\centering
 	\bgroup
 	\setlength\tabcolsep{-0.1cm}
 	\begin{tabular}{ccc}
 		\textsc{synthetic} & \textsc{isolet} & \textsc{mfeat} \\
 		\includegraphics[width=0.36\textwidth]{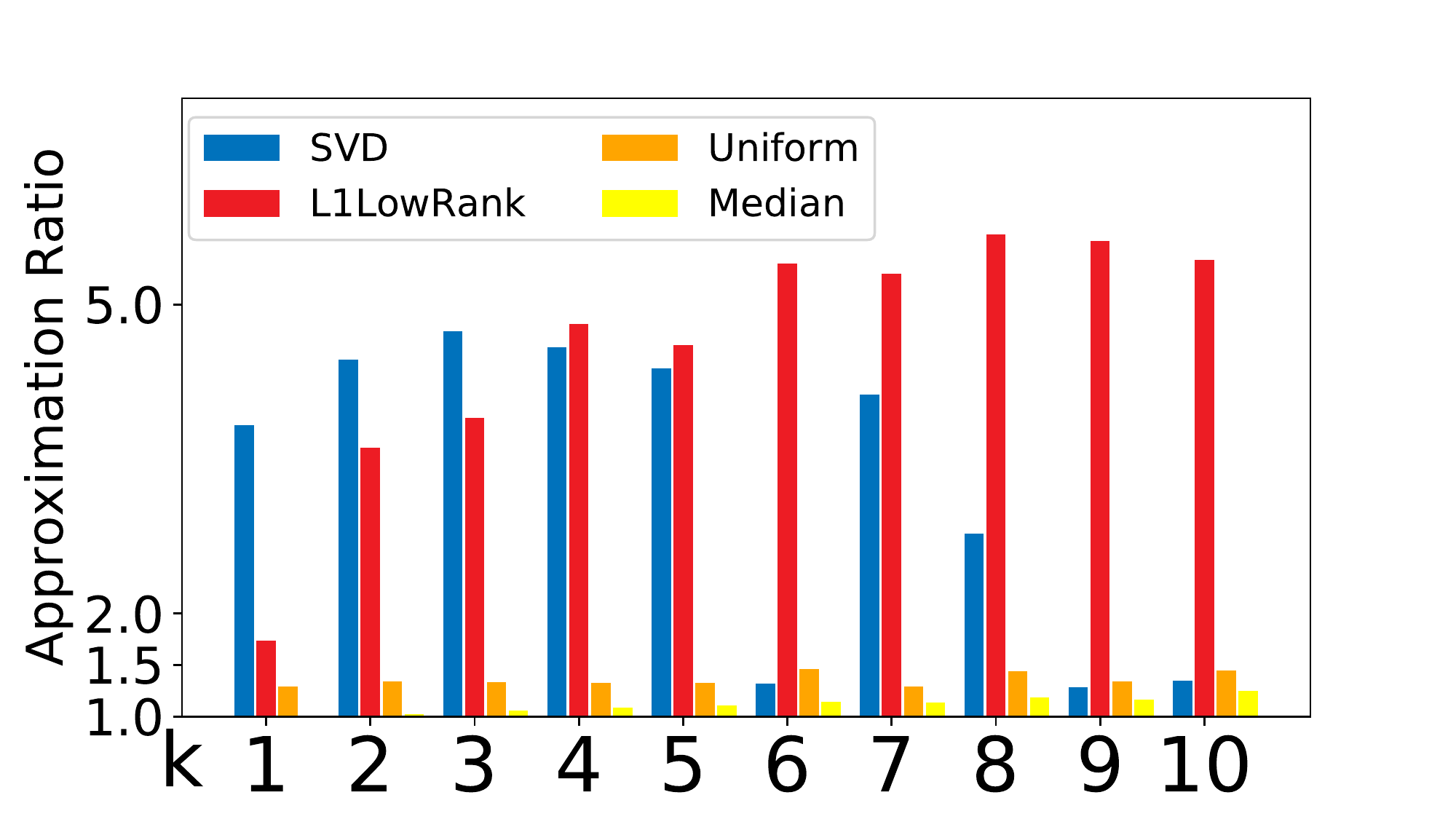}&
 		\includegraphics[width=0.36\textwidth]{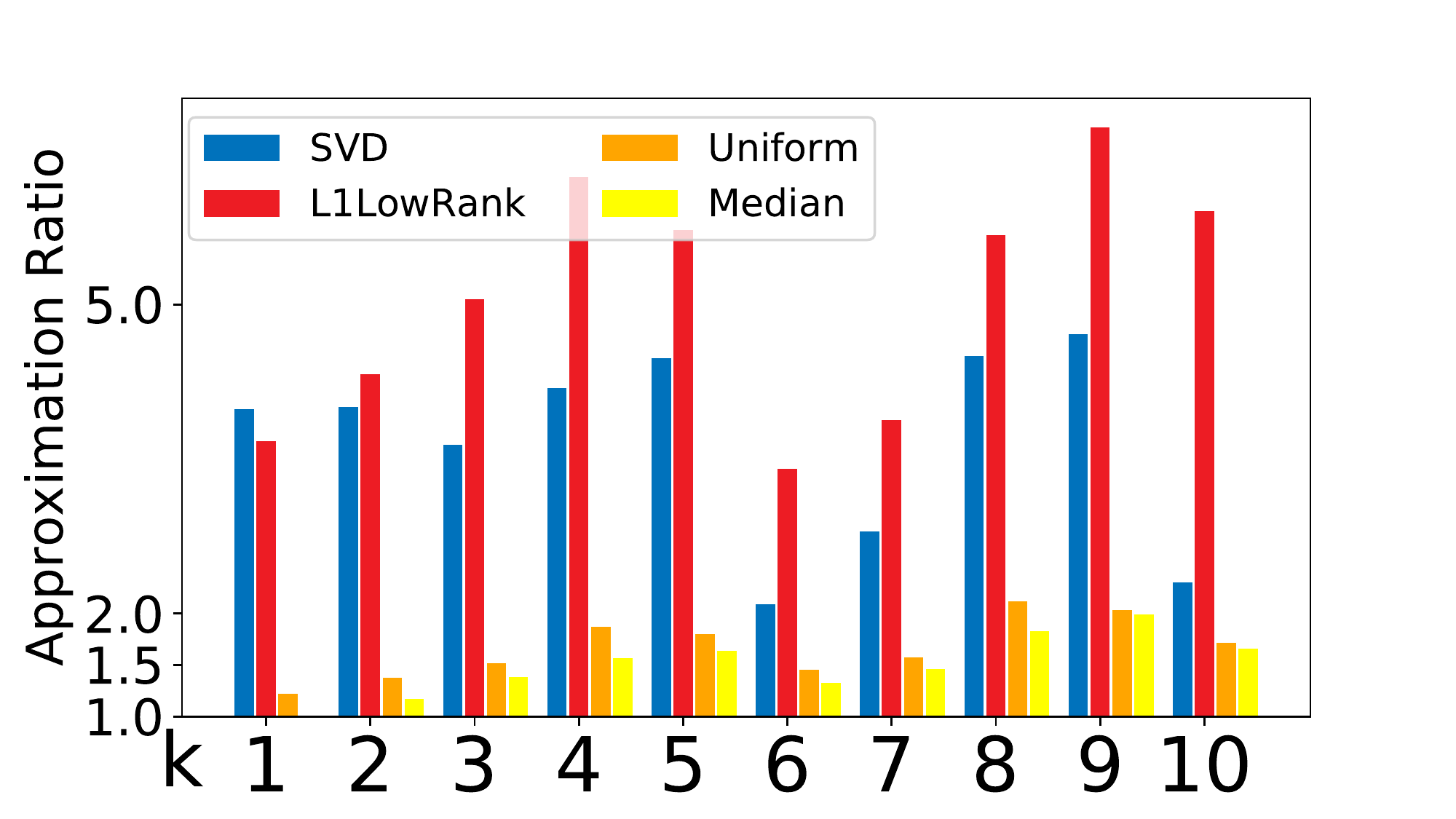}&
 		\includegraphics[width=0.36\textwidth]{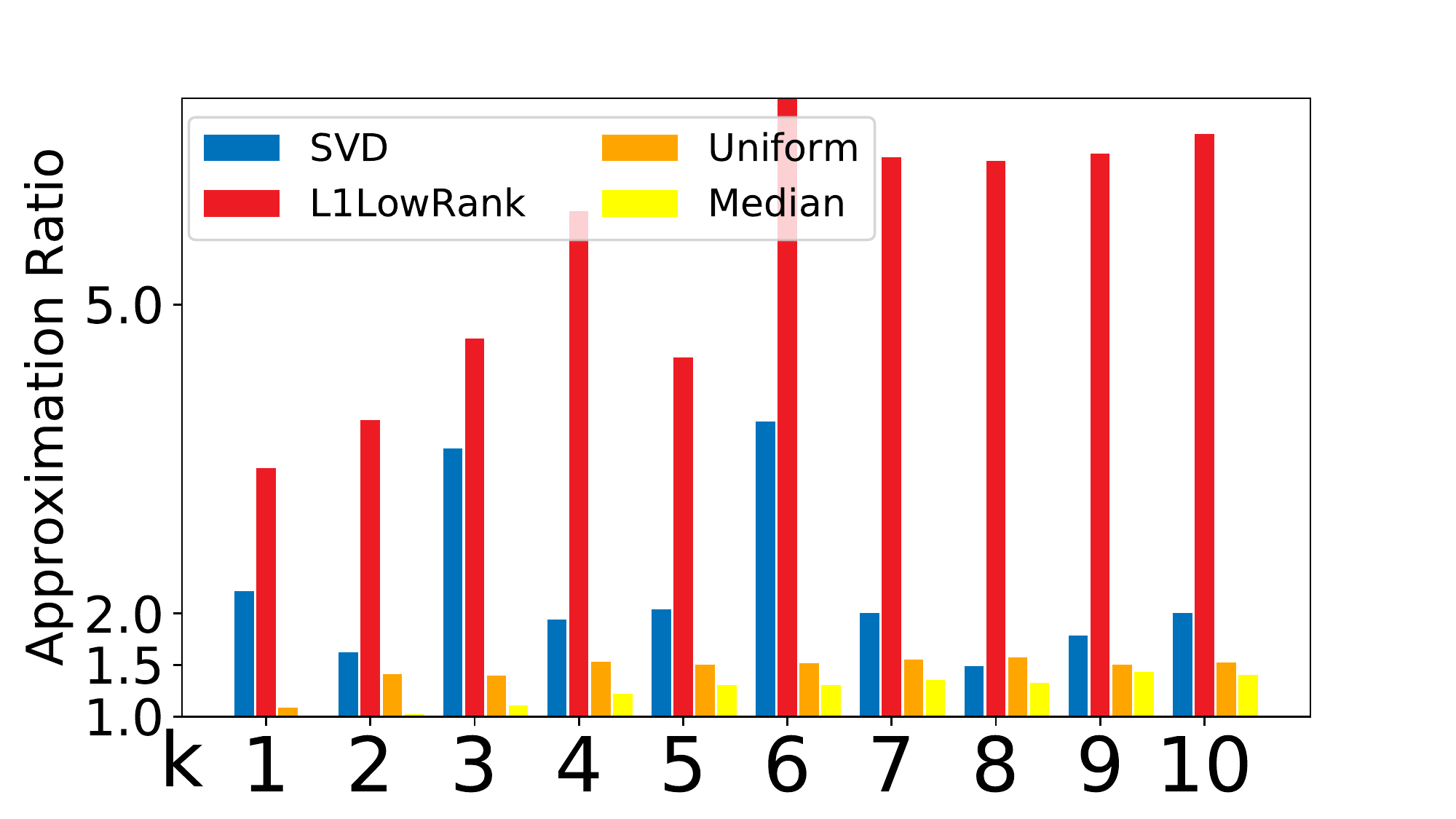}\\
 		\includegraphics[width=0.36\textwidth]{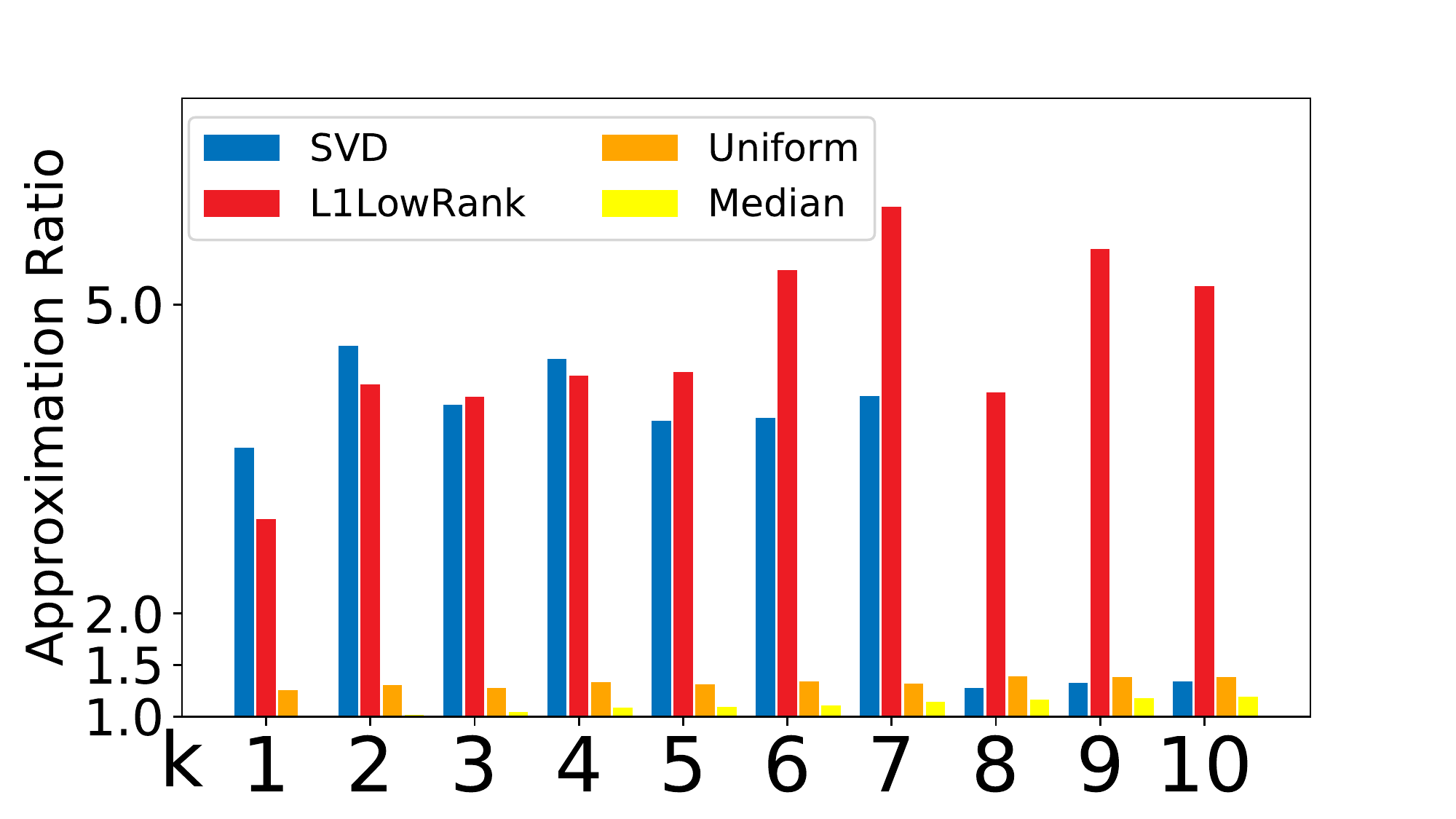}&
 		\includegraphics[width=0.36\textwidth]{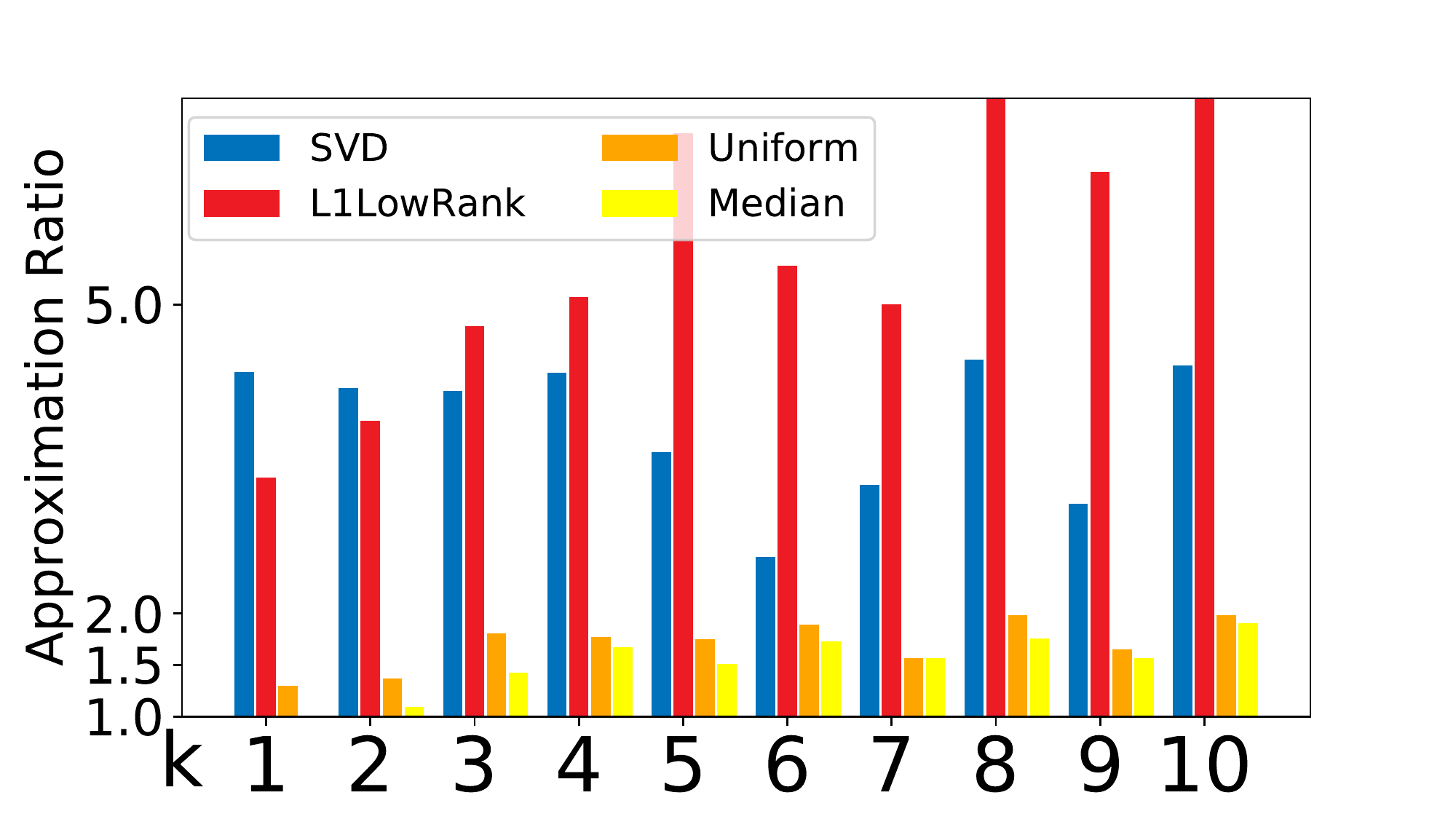}&
 		\includegraphics[width=0.36\textwidth]{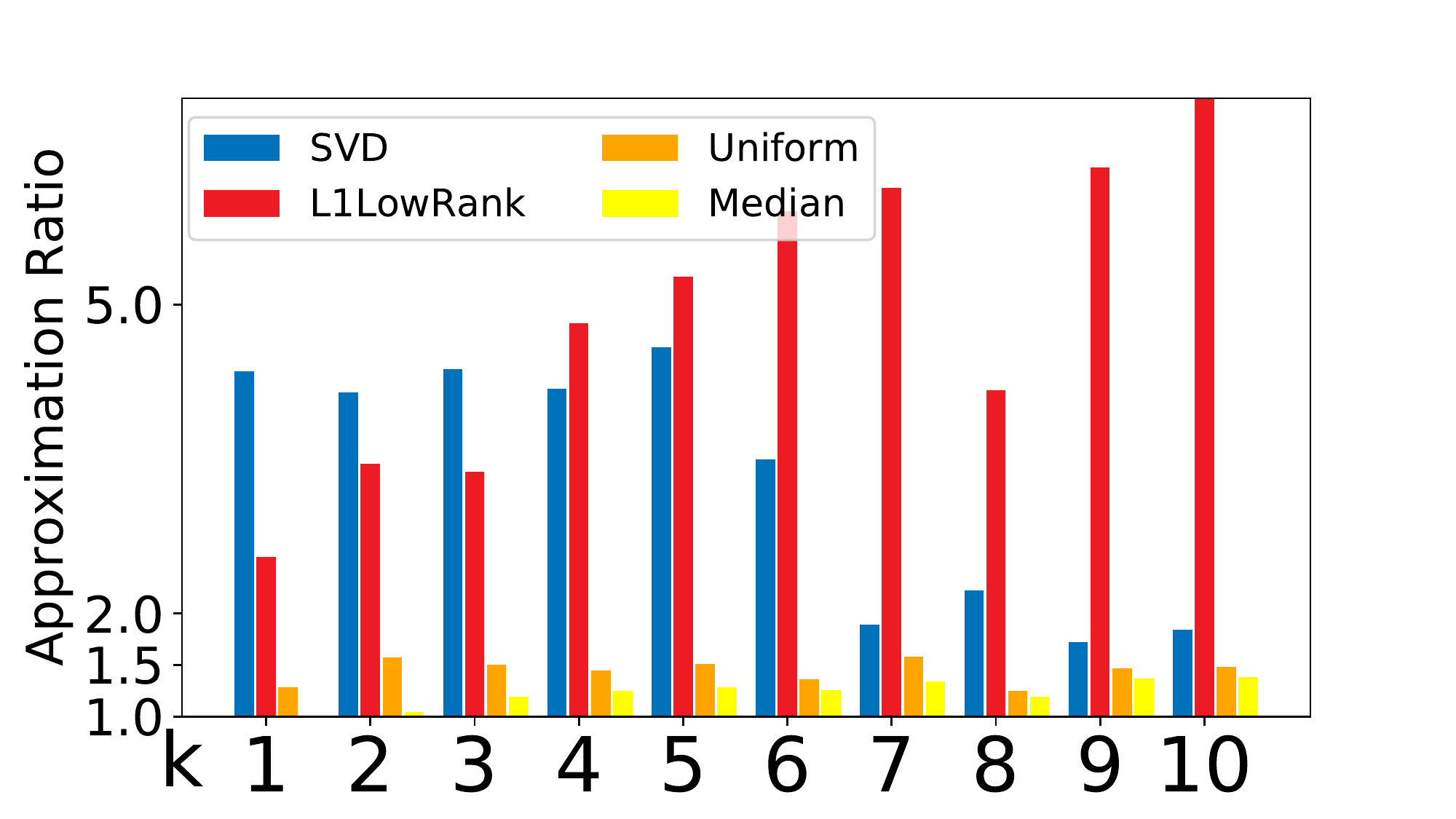}\\
 	\end{tabular}
 	\egroup
 	
 	\caption{\small \textbf{Empirical results.} The noise distributions of the experiments in the first row are from a $1.1$-stable distribution. The noise distributions corresponding to the second row are the $1.1$-th root of a Cauchy distribution. 
 	The blue, red, orange and yellow bar denote SVD, the entrywise $\ell_1$-norm low rank algorithm in \cite{swz17}, the uniform $k$-column subset sampling algorithm in \cite{cgklrw17}, and Algorithm~\ref{alg:heu} respectively.}\label{fig:results}
 \end{figure*}

\paragraph{Acknowledgments.}
David P. Woodruff was supported in part by Office of Naval Research (ONR) grant N00014- 18-1-2562. Part of this work was done while he was visiting the Simons Institute for the Theory of Computing.
Peilin Zhong is supported in part by NSF grants (CCF-1703925, CCF-1421161, CCF-1714818, CCF-1617955 and CCF-1740833), Simons Foundation (\#491119 to Alexandr Andoni), Google Research Award and a Google Ph.D. fellowship.
Part of this work was done 
while Zhao Song and Peilin Zhong were interns at IBM Research - Almaden and 
while Zhao Song was visiting the Simons Institute for the Theory of Computing.
\newpage
\ifdefined\isarxivversion
\bibliographystyle{alpha}
\bibliography{ref}
\else
\bibliographystyle{unsrt}
\bibliography{ref}
\fi
\newpage
\appendix
\section{Missing Proofs in Section \ref{sec:dis}}

\subsection{Proof of Lemma~\ref{lem:lower_bound_on_cost}}
\begin{proof}
	Let $Z\in\mathbb{R}^{n\times n}$ be a random matrix. For each $i,j\in[n],$ define random variable $Z_{i,j}$ as
	\begin{align*}
	Z_{i,j}=\left\{\begin{array}{ll} |\Delta_{i,j}|, & \text{~if~} |\Delta_{i,j}|\leq n; \\ n, &\mathrm{otherwise.}\end{array}\right.
	\end{align*}
	For $i,j\in [n],$ by Markov's inequality, we have
	\begin{align}
	\Pr[|\Delta_{i,j}|\geq n]=\Pr[|\Delta_{i,j}|^p\geq n^p]\leq \E[|\Delta_{i,j}|^p]/n^p=O(1/n^p).\label{eq:tail_bound_of_delta}
	\end{align}
	Notice that
	\begin{align*}
	\E[|\Delta_{i,j}|^p]=\int_{0}^n x^p f(x) \mathrm{d}x+\int_{n}^{\infty} x^p f(x) \mathrm{d}x=O(1)
	\end{align*}
	where $f(x)$ is the probability density function of $|\Delta_{i,j}|.$
	Thus we have
	\begin{align*}
	\int_{n}^{\infty} x f(x) \mathrm{d}x\leq\int_{n}^{\infty} x^p/n^{p-1}\cdot f(x) \mathrm{d}x=O(1/n^{p-1}).
	\end{align*}
	Because $\E[|\Delta_{i,j}|]=1,$ we have
	\begin{align}
	\int_{0}^{\infty} x f(x) \mathrm{d}x = \E[|\Delta_{i,j}|]-\int_{n}^{\infty} x f(x) \mathrm{d} x \geq 1-O(1/n^{p-1}).\label{eq:tail_bound_of_expectation}
	\end{align}
	By Equation~\eqref{eq:tail_bound_of_expectation}, we have
	\begin{align*}
	\E[Z_{i,j}]=\int_{0}^n x f(x) dx + n \cdot \Pr[|\Delta_{i,j}|\geq n]\geq \int_{0}^n x f(x) dx\geq 1-O(1/n^{p-1}).
	\end{align*}
	By Equation~\eqref{eq:tail_bound_of_delta} and $E[|\Delta_{i,j}|^p]\leq O(1)$, we have
	\begin{align*}
	\E[Z_{i,j}^2]=\int_{0}^n x^2 f(x) \mathrm{d}x+n^2\Pr[|\Delta_{i,j}|\geq n]\leq O(n^{2-p})+O(n^{2-p})=O(n^{2-p}).
	\end{align*}
	By the inequality of~\cite{m03},
	\begin{align*}
	\Pr[\E[\|Z\|_1]-\|Z\|_1\geq \epsilon \E[\|Z\|_1]/2]&\leq\exp\left(\frac{-\epsilon^2\E[\|Z\|_1]^2/4}{2\sum_{i,j}\E[Z_{i,j}^2]}\right)\\
	&\leq \exp\left(\frac{-\epsilon^2(n^2-O(n^{3-p}))^2/4}{2n^2\cdot O(n^{2-p})}\right)\\
	&\leq e^{-\Theta(n)}
	\end{align*}
	Thus with probability at least $1-e^{-\Theta(n)},$ $\|Z\|_1\geq (1-\epsilon/2)\E[\|Z\|_1]\geq (1-\epsilon)n^2$ where the last inequality follows by $\E[\|Z\|_1\geq n^2-O(n^{3-p})]$ and $1/\epsilon=n^{o(1)}.$ Since $\|\Delta\|_1\geq \|Z\|_1,$ we complete the proof.
\end{proof}

\subsection{Proof of Lemma~\ref{lem:averaging_works}}
\begin{proof}
	Let $Z\in\mathbb{R}^{n\times t}$ be a random matrix where $Z_{i,j}$ are i.i.d. random variables with probability density function:
	\begin{align*}
	g(x)=\left\{\begin{array}{ll}f(x)/\Pr[|\Delta_{1,1}|\leq n^{1/2+1/(2p)}], & \text{~if~} |x|\leq n^{1/2+1/(2p)}; \\ 0, & \mathrm{~otherwise.}\end{array}\right.
	\end{align*}
	where $f(x)$ is the probability density function of $\Delta_{1,1}$. (Note that in the above equation, $\Pr[|\Delta_{1,1}|\leq n^{1/2+1/(2p)}] > 0$.)  Now, we have $\forall a\geq 0$,
	\begin{align*}
	\Pr\left[\left\|\sum_{j=1}^t\alpha_j\Delta_j\right\|_1\leq a ~\bigg|~ \forall i\in[n],j\in[t],|\Delta_{i,j}|\leq n^{1/2+1/(2p)}\right]=\Pr\left[\left\|\sum_{j=1}^t\alpha_jZ_j\right\|_1\leq a\right].
	\end{align*}
	Now we look at the $i$-th row of $\sum_{j=1}^t\alpha_jZ_j.$ We have
	\begin{align}\label{eq:expectation_of_averaging}
	\E\left[\left|\sum_{j=1}^t\alpha_jZ_{i,j}\right|\right]
	= & ~ \left(\E\left[\left|\sum_{j=1}^t\alpha_jZ_{i,j}\right|\right]^p\right)^{1/p} \notag \\
	\leq & ~ \E\left[\left|\sum_{j=1}^t\alpha_jZ_{i,j}\right|^p\right]^{1/p} \notag\\
	\leq & ~ \E\left[\left(  \left( \sum_{j=1}^t\alpha_j^2Z_{i,j}^2 \right)^{1/2} \right)^p\right]^{1/p} \notag \\
	\leq & ~ \E\left[\sum_{j=1}^t|\alpha_jZ_{i,j}|^p\right]^{1/p} \notag\\
	\leq & ~ \left(\sum_{j=1}^t\E[|\alpha_jZ_{i,j}|^p]\right)^{1/p} \notag \\
	\leq & ~ \left(\sum_{j=1}^t\E[|Z_{i,j}|^p]\right)^{1/p} \notag\\
	\leq & ~ O(t^{1/p}),
	\end{align}
	where the first inequality follows by Jensen's inequality, the second inequality follows by Remark 3 of~\cite{l97}, the third inequality follows by $\|x\|_2\leq \|x\|_p$ for $p<2$, the fourth inequality follows by $|\alpha_j|\leq 1$, the fifth inequality follows by $\E[|Z_{i,j}|^p]=\E[|\Delta_{i,j}|^p\mid |\Delta_{1,1}|\leq n^{1/2+1/(2p)}]\leq \E[|\Delta_{i,j}|^p]=O(1).$
	For the second moment, we have
	\begin{align}\label{eq:the_bound_for_the_second_moment}
	\E\left[\left|\sum_{j=1}^t\alpha_jZ_{i,j}\right|^2\right]&=\sum_{j=1}^t\E\left[\alpha_j^2Z_{i,j}^2\right]+\sum_{j\not=k}\E[\alpha_j\alpha_kZ_{i,j}Z_{i,k}]\notag\\
	= & ~ \sum_{j=1}^t\alpha_j^2\E\left[Z_{i,j}^2\right]+\sum_{j\not=k}\alpha_j\alpha_k\E[Z_{i,j}]\E[Z_{i,k}]\notag\\
	\leq & ~ \sum_{j=1}^t\E\left[Z_{i,j}^2\right]\notag\\
	= & ~ t\cdot 2\int_{0}^{n^{1/2+1/(2p)}} x^2f(x)/\Pr \left[|\Delta_{i,j}|\leq n^{1/2+1/(2p)} \right]\mathrm{d}x\notag\\
	\leq & ~ 2t/\Pr \left[ |\Delta_{i,j}|\leq n^{1/2+1/(2p)} \right]\cdot (n^{1/2+1/(2p)})^{2-p}\int_{0}^{n^{1/2+1/(2p)}}x^pf(x)\mathrm{d}x\notag\\
	\leq & ~ O(tn^{2-p}),
	\end{align}
	where the second inequality follows by independence of $Z_{i,j}$ and $Z_{i,k}.$ The first inequality follows by $|\alpha_j|\leq 1$ and $\E[Z_{i,j}]=\E[Z_{i,k}]=0.$ The third equality follows by the probability density function of $Z_{i,j}.$ The second inequality follows by $x^{2-p}\leq   (n^{1/2+1/(2p)})^{2-p}$ when $0\leq x\leq  n^{1/2+1/(2p)}.$ The last inequality follows by $\E[|\Delta_{i,j}|^p]=O(1),p>1$ and $\Pr[|\Delta_{i,j}|\leq n^{1/2+1/(2p)}]\geq 1- \E[|\Delta_{i,j}|^p]/(n^{1/2+1/(2p)})^p=1-O(1/n^{p/2+1/2})\geq 1/2.$
	
	For $i\in[n],$ define $X_i=|\sum_{j=1}^t\alpha_jZ_{i,j}|.$ Then, by Bernstein's inequality
	\begin{align*}
	&\Pr\left[\left\|\sum_{j=1}^t\alpha_jZ_j\right\|_1-\E\left[\left\|\sum_{j=1}^t\alpha_jZ_j\right\|_1\right]\geq 0.5t^{1/p}n\right]\\
	=~&\Pr\left[\sum_{i=1}^n X_i-\E\left[\sum_{i=1}^n X_i\right]\geq 0.5t^{1/p}n\right]\\
	\leq~&\exp\left(-\frac{0.5\cdot0.5^2t^{2/p}n^2}{\sum_{i=1}^n \E[X_i^2]+\frac{1}{3}n^{1/2+1/(2p)}\cdot0.5t^{1/p}n}\right)\\
	\leq~&e^{-n^{\Theta(1)}}.
	\end{align*}
	The last inequality follows by Equation~\eqref{eq:the_bound_for_the_second_moment}. According to Equation~\eqref{eq:expectation_of_averaging}, with probability at least $1-e^{-n^{\Theta(1)}},$
	\begin{align*}
	\left\|\sum_{j=1}^t\alpha_jZ_j\right\|_1\leq \E\left[\left\|\sum_{j=1}^t\alpha_jZ_j\right\|_1\right]+0.5t^{1/p}n\leq O(t^{1/p}n).
	\end{align*}
\end{proof}

\subsection{Proof of Lemma~\ref{lem:hard_one_is_small}}
\begin{proof}
	For $i,j\in[n],$ we have
	\begin{align*}
	\Pr \left[ |\Delta_{i,j}|> n^{1/2+1/(2p)} \right]=\Pr \left[ |\Delta_{i,j}|^p> n^{p/2+1/2} \right]\leq \E \left[|\Delta_{i,j}|^p \right]/n^{p/2+1/2}\leq O(1/n^{p/2+1/2}).
	\end{align*}
	For column $j$, by taking a union bound,
	\begin{align*}
	\Pr[j\in H]=\Pr \left[\exists i\in[n],|\Delta_{i,j}|>n^{1/2+1/(2p)} \right]\leq O(1/n^{p/2-1/2}).
	\end{align*}
	Thus, $\E[|H|]\leq O(n^{1-(p-1)/2}).$ By applying Markov's inequality, we complete the proof.
\end{proof}

\subsection{Proof of Lemma~\ref{lem:remaining_is_small}}
\begin{proof}
	For $l\in\mathbb{N}_{\geq 0},$ define $G_l=\{j\mid \|\Delta_j\|_1\in(n\cdot2^l,n\cdot 2^{l+1}]\}.$ We have
	\begin{align*}
	\E[|G_l|] \leq & ~ \sum_{j=1}^n \Pr \left[ \|\Delta_{j}\|_1\geq n\cdot 2^l \right]\\
	= & ~ n\Pr \left[ \|\Delta_{1}\|_1\geq n\cdot 2^l \right]\\
	\leq & ~ n\Pr \left[n^{1-1/p}\|\Delta_{1}\|_p\geq n\cdot 2^l \right]\\
	= & ~ n\Pr \left[ n^{p-1}\|\Delta_{1}\|_p^p\geq n^p\cdot 2^{lp} \right]\\
	\leq & ~ n \E \left[ n^{p-1}\|\Delta_{1}\|_p^p \right]/(n^p\cdot 2^{lp})\\
	\leq & ~ O(n/2^{lp}).
	\end{align*}
	The first inequality follows by the definition of $G_l$. The second inequality follows since $\forall x\in \mathbb{R}^n,\|x\|_1\leq n^{1-1/p}\|x\|_p.$ The third inequality follows by Markov's inequality. The last inequality follows since $\forall i,j\in[n],\E[|\Delta_{i,j}|^p]=O(1).$
	
	Let $l^*\in\mathbb{N}_{\geq 0}$ satisfy $2^{l^*}< \epsilon r$ and $2^{l^*+1}\geq \epsilon r.$ We have
	\begin{align*}
	\E\left[\sum_{j:\|\Delta_j\|_1\geq n2^{l^*}}\|\Delta_j\|_1\right]&\leq \E\left[\sum_{l=l^*}^\infty|G_l|\cdot n2^{l+1}\right]
	=\sum_{l=l^*}^{\infty}\E[|G_l|]\cdot n2^{l+1}\\
	&\leq \sum_{l=l^*}^{\infty} O(n/2^{lp})\cdot n2^{l+1}
	= \sum_{l=l^*}^{\infty} O(n^2/2^{l(p-1)})\\
	&= O(n^2/2^{l^*(p-1)})
	= O(n^2/(\epsilon r)^{p-1})\\
	&=O(\epsilon n^2).
	\end{align*}
	By Markov's inequality, with probability at least $.999,$ $\sum_{j:\|\Delta_j\|_1\geq n2^{l^*}}\|\Delta_j\|_1\leq O(\epsilon n^2).$ Conditioned on $\sum_{j:\|\Delta_j\|_1\geq n2^{l^*}}\|\Delta_j\|_1\leq O(\epsilon n^2),$ for any $S\subset[n]$ with $|S|\leq n/r,$ we have
	\begin{align*}
	\sum_{j\in S}\|\Delta_j\|_1\leq |S|\cdot n2^{l^*}+\sum_{j:\|\Delta_j\|_1\geq n2^{l^*}}\|\Delta_j\|_1\leq \epsilon n^2+O(\epsilon n^2)=O(\epsilon n^2).
	\end{align*}
	The second inequality follows because $|S|\leq n/r,2^{l^*}\leq \epsilon r$ and $\sum_{j:\|\Delta_j\|_1\geq n2^{l^*}}\|\Delta_j\|_1\leq O(\epsilon n^2).$
\end{proof}

\subsection{Proof of Lemma~\ref{lem:easy_one_is_concentrated}}
\begin{proof}
	Let $M=n^{1/2+1/(2p)}.$ Let $Z\in\mathbb{R}^{n}$ be a random vector where $Z_{i}$ are i.i.d. random variables with probability density function
	\begin{align*}
	g(x)=\left\{\begin{array}{ll}f(x)/\Pr[|\Delta_1|\leq M]& \text{~if~} 0\leq x\leq M ; \\ 0 & \text{~otherwise.}\end{array}\right.
	\end{align*}
	where $f(x)$ is the probability density function of $|\Delta_1|.$ Then $\forall a>0$
	\begin{align*}
	\Pr\left[\|\Delta\|_1\leq a\mid \forall i\in[n],|\Delta_i|\leq M\right]=\Pr\left[\|Z\|_1\leq a\right].
	\end{align*}
	For $i\in [n],$ because $\E[|\Delta_i|]=1,$ it holds that $\E[Z_i]\leq 1.$ We have $\E[\sum_{i=1}^n Z_i]\leq n.$ For the second moment, we have
	\begin{align*}
	\E[Z_i^2]&=\int_{0}^{M} x^2f(x)/\Pr[|\Delta_1|\leq M]\mathrm{d}x\\
	&\leq M^{2-p}/\Pr[|\Delta_1|\leq M]\int_{0}^{M} x^p f(x)\mathrm{d}x\\
	&\leq O(M^{2-p})\\
	&\leq O(n^{2-p})
	\end{align*}
	where the second inequality follows by $\E[|\Delta_1|^p]=O(1),$ and $\Pr[|\Delta_1|\leq M]\geq 1-\E[|\Delta_1|^p]/M^p\geq 1/2.$
	
	Then by Bernstein's inequality, we have
	\begin{align*}
	&\Pr\left[\sum_{i=1}^n Z_i-E\left[\sum_{i=1}^n Z_i\right]\geq\epsilon n\right]\\
	\leq~&\exp\left(\frac{-0.5\epsilon^2n^2}{\sum_{i=1}^n \E[Z_i^2]+\frac{1}{3}M\cdot\epsilon n}\right)\\
	\leq~&e^{-n^{\Theta(1)}}.
	\end{align*}
	Thus,
	\begin{align*}
	\Pr\left[\|\Delta\|_1\leq (1+\epsilon)n\mid \forall i\in[n],|\Delta_i|\leq M\right]=\Pr\left[\|Z\|_1\leq (1+\epsilon)n\right]\geq 1-e^{-n^{\Theta(1)}}.
	\end{align*}
\end{proof}

\subsection{Proof of Lemma~\ref{lem:good_tuple_low_cost}}
\begin{proof}
	Recall that $(S_1,S_2,\cdots,S_t,i)$ is equivalent to $(S_{[t]},i)$. Let $(S_{[t]},i)$ be an $(A^*,q,t,1/2)$-good tuple which satisfies $H\cap\left(\bigcup_{j=1}^t S_j\right)=\emptyset.$ Let $C$ be the core of $(S_{[t]},i).$ Let $(x_1,x_2,\cdots,x_t)$ be the coefficients tuple corresponding to $(S_{[t]},i).$ Then we have that
	\begin{align*}
	&\left\|\frac{1}{|C|}\sum_{j=1}^t A_{S_j}x_j-A_i\right\|_1\\
	=~&\left\|\frac{1}{|C|}\sum_{j=1}^t \left(A^*_{S_j}+\Delta_{S_j}\right)x_j-(A^*_i+\Delta_i)\right\|_1\\
	\leq~&\left\|\frac{1}{|C|}\sum_{j=1}^t A^*_{S_j}x_j-A^*_i\right\|_1+\|\Delta_i\|_1+\frac{1}{|C|}\left\|\sum_{j=1}^t \Delta_{S_j}x_j\right\|_1\\
	=~&\|\Delta_i\|_1+\frac{1}{|C|}\left\|\sum_{j=1}^t \Delta_{S_j}x_j\right\|_1\\
	\leq~&\|\Delta_i\|_1+\frac{2}{t}\left\|\sum_{j=1}^t \Delta_{S_j}x_j\right\|_1\\
	\leq~&\|\Delta_i\|_1+O\left(\frac{1}{t}\cdot(qt)^{1/p}n\right)\\
	=~&\|\Delta_i\|_1+O\left(q^{1/p}/t^{1-1/p}n\right)
	\end{align*}
	holds with probability at least $1-2^{-n^{\Theta(1)}}.$ The first equality follows using $A=A^*+\Delta.$ The first inequality follows using the triangle inequality. The second equality follows using the definition of the core and the coefficients tuple (see Definition~\ref{def:core_and_good_tuple} and Definition~\ref{def:coefficients_tuple}). The second inequality follows using Definition~\ref{def:core_and_good_tuple}. The third inequality follows by Lemma~\ref{lem:averaging_works} and the condition that $H\cap\left(\bigcup_{j=1}^t S_j\right)=\emptyset.$
	
	Since the size of $\left|\{i\}\cup\left(\bigcup_{j=1}^t S_j\right)\right|=qt+1,$ the total number of $(A^*,q,t,1/2)-$good tuples is upper bounded by $n^{qt+1}\leq 2^{n^{o(1)}}.$ By taking a union bound, we complete the proof.
\end{proof}

\subsection{Proof of Lemma~\ref{lem:label_uniform_samples}}
\begin{proof}
	For $j\in[t],$ by symmetry of the choices of $S_j$ and $i$, we have $\Pr[i\in R_{A^*}(S_j\cup\{i\})]\leq k/(q+1).$ Thus, by Markov's inequality,
	\begin{align*}
	&\Pr[|\{j\in[t]\mid i\in R_{A^*}(S_j\cup\{i\})\}|>0.5t]\\
	\leq~& \E[|\{j\in[t]\mid i\in R_{A^*}(S_j\cup\{i\})\}|]/(0.5t)\\
	\leq~&2k/q.
	\end{align*}
	Thus,
	\begin{align*}
	&\Pr[|\{j\in[t]\mid i\not\in R_{A^*}(S_j\cup\{i\})\}|\geq0.5t]\geq 1-2k/q.
	\end{align*}
\end{proof}

\subsection{Proof of Lemma~\ref{lem:easy_to_find_good_tuple}}
\begin{proof}
	For $S_1,S_2,\cdots,S_t\in {[n]\choose q}$ with $\sum_{j=1}^t |S_j|=qt,$ define
	\begin{align*}
	P_{(S_1,S_2,\cdots,S_t)}=\Pr_{i\in [n]\setminus \left(\bigcup_{j=1}^t S_j\right)}[(S_1,S_2,\cdots,S_t,i)\text{~is an $(A^*,q,t,1/2)-$good tuple~}].
	\end{align*}
	Let set $T$ be defined as follows:
	\begin{align*}
	\left\{ (S_1,S_2,\cdots,S_t) ~\bigg|~ S_1,S_2,\cdots,S_t\in {[n]\choose q} \text{~with~} \sum_{j=1}^t |S_j|=qt \right\}.
	\end{align*}
	Let $G$ be the set of all the $(A^*,q,t,1/2)-$good tuples. Then, we have
	\begin{align*}
	&\Pr_{(S_1,S_2,\cdots,S_t)\sim T}\left[\left|\left\{i\in[n]\setminus \left(\cup_{j=1}^t S_j\right)\mid (S_1,S_2,\cdots,S_t,i)\in G\right\}\right|\geq (1-4k/q)(n-qt)\right]\\
	=~&\frac{1}{|T|}\left|\left\{(S_1,S_2,\cdots,S_t)\mid (S_1,S_2,\cdots,S_t)\in T\text{~and~}P_{(S_1,S_2,\cdots,S_t)}\geq 1-4k/q\right\}\right|\\
	=~&\frac{1}{|T|}\underset{P_{(S_1,S_2,\cdots,S_t)}\geq 1-4k/q}{\sum_{(S_1,S_2,\cdots,S_t)\in T}} 1\\
	\geq~& \frac{1}{|T|}\underset{P_{(S_1,S_2,\cdots,S_t)}\geq 1-4k/q}{\sum_{(S_1,S_2,\cdots,S_t)\in T}} P_{(S_1,S_2,\cdots,S_t)}\\
	\geq~&1-2k/q-\frac{1}{|T|}\underset{P_{(S_1,S_2,\cdots,S_t)}< 1-4k/q}{\sum_{(S_1,S_2,\cdots,S_t)\in T}} P_{(S_1,S_2,\cdots,S_t)}\\
	\geq~&1-2k/q-(1-4k/q)\\
	\geq~&2k/q.
	\end{align*}
	The second inequality follows from Lemma~\ref{lem:label_uniform_samples}
	\begin{align*}
	\frac{1}{|T|}\underset{P_{(S_1,S_2,\cdots,S_t)}< 1-4k/q}{\sum_{(S_1,S_2,\cdots,S_t)\in T}} P_{(S_1,S_2,\cdots,S_t)}+\frac{1}{|T|}\underset{P_{(S_1,S_2,\cdots,S_t)}\geq 1-4k/q}{\sum_{(S_1,S_2,\cdots,S_t)\in T}} P_{(S_1,S_2,\cdots,S_t)}\geq 1-2k/q.
	\end{align*}
\end{proof}

\section{Hardness Result}\label{sec:hard}

\paragraph{An overview of the hardness result.}
Recall that we overcame the column subset selection lower bound
of \cite{swz17}, which shows for entrywise $\ell_1$-low rank approximation
that there are matrices for which any subset of $\poly(k)$
columns spans at best a $k^{\Omega(1)}$-approximation. Indeed, we came up with a column subset
of size $\poly(k(\epsilon^{-1} + \log n))$ spanning a $(1+\epsilon)$-approximation. To do this,
we assumed $A = A^* + \Delta$, where $A^*$ is an arbitrary rank-$k$ matrix, and the entries
are i.i.d. from a distribution with $\E[|\Delta_{i,j}|] = 1$
and $\E[|\Delta_{i,j}|^p] = O(1)$ for any real number $p$ strictly greater than
$1$.

Here we show an assumption on the moments is necessary,
by showing if instead $\Delta$ were drawn from a matrix of i.i.d. Cauchy random variables, for which the
$p$-th moment is undefined or infinite for all $p \geq 1$,
then for any subset of $n^{o(1)}$ columns, it spans at best a $1.002$ approximation. The input
matrix $A = n^C 1 \cdot 1^\top  + \Delta$, where $C > 0$ is a constant and we show that $n^{\Omega(1)}$
columns need to be chosen to obtain a $1.001$-approximation, even for $k = 1$.
Note that this
result is stronger than that in \cite{swz17} in that it rules out column subset selection even if one
were to choose $n^{o(1)}$ columns; the result in \cite{swz17} requires at most $\poly(k)$ columns, which
for $k = 1$, would just rule out $O(1)$ columns. Our
main goal here is to show that a moment assumption on our distribution is necessary, and our
result also applies to a symmetric noise distribution which is i.i.d. on all entries, whereas
the result of \cite{swz17} requires a specific deterministic pattern (namely, the identity matrix)
on certain entries.

Our main theorem is given in Theorem \ref{thm:dis_l1_hardness}. The outline of the proof is as follows. We first condition on the event that $\|\Delta\|_1 \leq \frac{4.0002}{\pi} n^2 \ln n$, which is shown in Lemma \ref{lem:bound_of_Delta} and follows form standard analysis of sums of absolute values of Cauchy random variables. Thus, it is sufficient to show if we choose any subset $S$ of $r = n^{o(1)}$ columns, denoted by the submatrix $A_S$, then
$\min_{X \in \mathbb{R}^{r \times n}}\|A_SX-A\|_1 \geq \frac{4.01}{\pi} \cdot n^2 \ln n$, as indeed then $\min_{X \in \mathbb{R}^{r \times n}}\|A_SX-A\|_1 \geq 1.002 \|\Delta\|_1$ and we rule out a $(1+\epsilon)$-approximation for $\epsilon$ a sufficiently small constant. To this end, we instead show for a fixed $S$, that $\min_{X \in \mathbb{R}^{r \times n}}\|A_SX-A\|_1 \geq \frac{4.01}{\pi} \cdot n^2 \ln n$ with probability $1-2^{-n^{\Theta(1)}}$, and then apply a union bound over all $S$. To prove this for a single subset $S$, we argue that for every ``coefficient matrix'' $X$, that
$\|A_SX-A\|_1 \geq \frac{4.01}{\pi} \cdot n^2 \ln n$.

We show in Lemma \ref{lem:for_all_alpha_can_not_be_too_large}, that with probability $1-(1/n)^{\Theta(n)}$ over $\Delta$, simultaneously for all $X$,
if $X$ has a column $X_j$ with $\|X_j\|_1 \geq n^c$ for a constant $c >0$,
then $\|A_SX_j - A_j\|_1 \geq .9n^3$, which is already too large to provide an $O(1)$-approximation.
Note that we need such a high probability bound to later union bound over {\it all $S$}.
Lemma \ref{lem:for_all_alpha_can_not_be_too_large} is in turn shown via a net argument on all $X_j$ (it suffices to prove this for a single $j \in [n]$,
since there are only $n$ different
$j$, so we can union bound over all $j$). The net bounds are given in Definition \ref{def:l1epsnet} and Definition \ref{lem:l1_eps_net_size},
and the high probability bound for a given coefficient vector $X_j$ is shown in Lemma \ref{lem:for_each_alpha_can_not_be_too_large}, where we use properties
of the Cauchy distribution. Thus, we can assume $\|X_j\|_1 < n^c$ for all $j \in [n]$. We also show in Fact \ref{fac:bound_of_sum_of_alpha},
conditioned on the fact that $\|\Delta\|_1 \leq \frac{4.002}{\pi} n^2 \ln n$,
it holds that for {\it any} vector $X_j$, if $\|X_j\|_1 < n^c$ and $|1-{\bf 1}^\top X_j| > 1-10^{-20}$,
then $\|A_SX-A\|_1 \geq \|A_SX_j - A_j\|_1 > n^3$. The intuition here is $A = n^{c_0}{\bf 1} \cdot {\bf 1^\top} + \Delta$
for a large constant $c_0$, and $X_j$ does not have enough norm ($\|X_j\|_1 \leq n^c$) or correlation with the vector ${\bf 1}$
($|1-{\bf 1}^\top X_j| > 1-10^{-20}$) to make $\|A_SX_j - A_j\|_1$ small.

Given the above, we can assume both that
$\|X_j\|_1 \leq n^c$ and $|1-{\bf 1}^\top X_j| \leq 1-10^{-20}$ for all columns $j$ of our coefficient matrix $X$. We can also assume
that $\|A_SX-A\|_1 \leq 4n^2 \ln n$, as otherwise such an $X$ already satisfies $\|A_SX-A\|_1 \geq \frac{4.01}{\pi} \cdot n^2 \ln n$ and
we are done.
To analyze $\|A_SX-A\|_1 = \sum_{i,j} |(A_SX-A_{[n]\setminus S})_{i,j}|$
in Theorem \ref{thm:dis_l1_hardness}, we then split the sum over ``large coordinates'' $(i,j)$ for which
$|\Delta_{i,j}|> n^{1.0002}$, and ``small coordinates'' $(i,j)$ for which $|\Delta_{i,j}| < n^{.9999}$, and since we
seek to lower bound $\|A_SX-A_{[n] \setminus S}\|_1$, we drop the remaining coordinates $(i,j)$. To handle large coordinates, we observe that since the column span of
$A_S$ is only $r = n^{o(1)}$-dimensional, as one ranges over all vectors $y$ in its span of $1$-norm, say, $O(n^2 \ln n)$, there is only a small subset
$T$, of size at most $n^{.99999}$ of coordinates $i \in [n]$ for which we could ever have $|y_i| \geq n^{1.0001}$.
We show this in Lemma \ref{lem:size_of_the_bad_region}. This uses the property of vectors in low-dimensional subspaces, and has been exploited in earlier works in the context of designing so-called subspace embeddings
\cite{cw13,mm13}. We call $T$ the ``bad region'' for $A_S$.
While the column span of $A_S$ depends on $\Delta_S$, it is independent of $\Delta_{[n] \setminus S}$, and thus it is extremely
unlikely that the large coordinate of $\Delta_S$ ``match up'' with the bad region of $A_S$. This is captured in Lemma \ref{lem:bound_for_large_part},
where we show that if $\|A_SX-A_{[n] \setminus S}\|_1 \leq 4n^2 \ln n$ (as we said we could assume above),
then $\sum_{\textrm{large coordinates }i,j}|(A_SX-A_{[n] \setminus S})_{i,j}|$ is at least $\frac{1.996}{\pi}n^2 \ln n$.
Intuitively, the heavy coordinates make up about
$\frac{2}{\pi}n^2 \ln n$ of the total mass of $\|\Delta\|_1$, by tail bounds of the Cauchy distribution, and for any set $S$ of size $n^{o(1)}$,
$A_S$ fits at most a small portion of this, still leaving us left with $\frac{1.996}{\pi} n^2 \ln n$ in cost. Our goal is to show
that $\|A_SX-A_{[n] \setminus S}\|_1 \geq \frac{4.01}{\pi} \cdot n^2 \ln n$, so we still have a way to go.

We next analyze $\sum_{\textrm{small coordinates }i,j}|(A_SX-A_{[n] \setminus S})_{i,j}|$. Via Bernstein's inequality, in Lemma \ref{lem:small_part_same_sign_delta} we argue
that for any fixed vector $y$ and random vector $\Delta_j$ of i.i.d. Cauchy entries, roughly half of the contribution of coordinates to
$\|\Delta_j\|_1$ will come from coordinates
$j$ for which sign$(y_j) = $sign$(\Delta_j)$ and $|\Delta_j| \leq n^{.9999}$, giving us a contribution of roughly $\frac{.9998}{\pi} n \ln n$
to the cost. The situation we will actually be in, when analyzing a
column of $A_SX-A_{[n] \setminus S}$, is that of taking the sum of two independent Cauchy vectors, shifted by a multiple of ${\bf 1}^\top$.
We
analyze this setting in Lemma \ref{lem:small_part_same_sign_alpha}, after first conditioning on certain level sets having typical behavior in Lemma \ref{lem:bound_of_Cauchy_level_size}.
This roughly doubles the contribution, gives us roughly a contribution of
$\frac{1.996}{\pi} n^2 \ln n$ from coordinates $j$ for which $(i,j)$ is a small coordinate and we look at coordinates $i$ on which the
sum of two independent Cauchy vectors have the same sign. Combined with the contribution from the heavy coordinates, this gives
us a cost of roughly $\frac{3.992}{\pi} n^2 \ln n$, which still falls short of the $\frac{4.01}{\pi} \cdot n^2 \ln n$ total cost
we are aiming for. Finally, if we sum up two independent Cauchy vectors and look at the contribution to the sum from coordinates which
disagree in sign, due to the anti-concentration of the Cauchy distribution we can still ``gain a little bit of cost'' since the values, although differing in sign, are still likely not to be very close in magnitude. We formalize this
in Lemma \ref{lem:small_part_different_sign}. We combine all of the costs from small coordinates in Lemma \ref{lem:fixed_small_part}, where we show we obtain a contribution
of at least $\frac{2.025}{\pi} n \ln n$. This is enough, when combined with our earlier $\frac{1.996}{\pi} n^2 \ln n$ contribution
from the heavy coordinates, to obtain an overall $\frac{4.01}{\pi} \cdot n^2 \ln n$ lower bound on the cost, and conclude the proof
of our main theorem in Theorem \ref{thm:dis_l1_hardness}.

In the remaining sections, we will present detailed proofs.
\subsection{A Useful Fact}

\begin{fact}\label{fac:bound_of_sum_of_alpha}
Let $c_0>0$ be a sufficiently large constant. Let $u= n^{c_0} \cdot {\bf 1 } \in \R^n$ and $\Delta\in\mathbb{R}^{n\times (d+1)}.$ If $\sum_{i=1}^{d+1} \| \Delta_{i}\|_1 \leq n^3$ and if $\alpha \in \R^d$ satisfies $ | 1 - {\bf 1}^\top \alpha | > 1/n^{c_1}$ and $\|\alpha\|_1\leq n^c$, where $0<c<c_0-10$ is a constant and $c_1>3$ is another constant depending on $c_0,c$, then
\begin{align*}
\| u - u {\bf 1}^\top \alpha + \Delta_{d+1} - \Delta_{[d]} \alpha \|_1 > n^3.
\end{align*}
\end{fact}
\begin{proof}
\begin{align*}
&\| u - u {\bf 1}^\top \alpha + \Delta_{d+1} - \Delta_{[d]} \alpha \|_1\\
\geq ~& |1-{\bf 1}^\top \alpha|\cdot\|u\|_1-\|\Delta_{d+1}\|_1-\|\Delta_{[d]}\alpha\|_1\\
\geq ~& |1-{\bf 1}^\top \alpha|\cdot n\cdot n^{c_0}-n^3-n^4\|\alpha\|_1\\
\geq~ & |1-{\bf 1}^\top \alpha|\cdot n\cdot n^{c_0}-n^{5+c}\\
\geq~ & n^{c_0+1-c_1}-n^{5+c}\\
\geq~ & n^3.
\end{align*}
The first inequality follows by the triangle inequality. The second inequality follows since $u= n^{c_0} \cdot {\bf 1 } \in \R^n$ and $\sum_{i=1}^{d+1} \| \Delta_{i}\|_1 \leq n^3.$ The third inequality follows since $\|\alpha\|_1\leq n^c.$ The fourth inequality follows since $ | 1 - {\bf 1}^\top \alpha | > 1/n^{c_1}.$ The last inequality follows since $c_0-c_1>c+5.$
\end{proof}

\subsection{One-Sided Error Concentration Bound for a Random Cauchy Matrix}

\begin{lemma}[Lower bound on the cost]\label{lem:bound_of_Delta} If $n$ is sufficiently large, then 
\begin{align*}
\Pr_{\Delta \sim \{C(0,1)\}^{n\times n}} \left[ \| \Delta \|_1 \leq \frac{4.0002}{\pi} n^2 \ln n \right] \geq 1-O(1/\log\log n).
\end{align*}
\end{lemma}
\begin{proof}
Let $\Delta\in\mathbb{R}^{n\times n}$ be a random matrix such that each entry is an i.i.d. $C(0,1)$ random Cauchy variable. Let $B=n^2\ln\ln n.$ Let $Z\in\mathbb{R}^{n\times n}$ and $\forall i,j\in[n],$
\begin{align*}
Z_{i,j}=\left\{\begin{array}{ll}|\Delta_{i,j}|&|\Delta_{i,j}|<B\\ B & \text{Otherwise}\end{array}\right..
\end{align*}
For fixed $i,j\in[n],$ we have
\begin{align*}
\E[Z_{i,j}]&=\frac{2}{\pi}\int_{0}^B \frac{x}{1+x^2} \mathrm{d}x+\Pr[|\Delta_{i,j}|\geq B]\cdot B\\
&= \frac{1}{\pi}\ln(B^2+1)+\Pr[|\Delta_{i,j}|\geq B]\cdot B\\
&\leq \frac{1}{\pi}\ln(B^2+1)+1
\end{align*}
where the first inequality follows by the cumulative distribution function of a half Cauchy random variable. We also have $\E[Z_{i,j}]\geq \frac{1}{\pi}\ln(B^2+1).$
For the second moment, we have
\begin{align*}
\E[Z_{i,j}^2]&=\frac{2}{\pi}\int_{0}^B \frac{x^2}{1+x^2} \mathrm{d}x+\Pr[|\Delta_{i,j}|\geq B]\cdot B^2\\
&=\frac{2}{\pi}(B-\tan^{-1} B)+\Pr[|\Delta_{i,j}|\geq B]\cdot B^2\\
&\leq \frac{2}{\pi} B + B\\
&\leq 2B
\end{align*}
where the first inequality follows by the cumulative distribution function of a half Cauchy random variable.
By applying Bernstein's inequality, we have
\begin{align}
&\Pr\left[\|Z\|_1-\E[\|Z\|_1]>0.0001 \E[\|Z\|_1]\right]\notag\\
\leq~&\exp\left(-\frac{0.5\cdot0.0001^2 \E[\|Z\|_1]^2}{n^2\cdot 2B+\frac{1}{3}B\cdot 0.0001 \E[\|Z\|_1]}\right)\notag\\
\leq~&\exp(-\Omega(\ln n/\ln\ln n))\notag\\
\leq~&O(1/\ln n).\label{eq:bernstein_sum_of_Z}
\end{align}
The first inequality follows by the definition of $Z$ and the second moment of $Z_{i,j}.$ The second inequality follows from $\E[\|Z\|_1]=\Theta(n^2\ln n)$ and $B=\Theta(n^2\ln\ln n).$
Notice that
\begin{align*}
~&\Pr \left[ \| \Delta \|_1 > \frac{4.0002}{\pi} n^2 \ln n \right]\\
=~&\Pr \left[ \| \Delta \|_1 > \frac{4.0002}{\pi} n^2 \ln n \mid \forall i,j, |\Delta_{i,j}|<B \right]\Pr\left[\forall i,j, |\Delta_{i,j}|<B\right]\\
~&+\Pr \left[ \| \Delta \|_1 > \frac{4.0002}{\pi} n^2 \ln n \mid \exists i,j, |\Delta_{i,j}|\geq B \right]\Pr\left[\exists i,j, |\Delta_{i,j}|\geq B\right]\\
\leq~&\Pr \left[ \| \Delta \|_1 > \frac{4.0002}{\pi} n^2 \ln n \mid \forall i,j, |\Delta_{i,j}|<B \right]+\Pr\left[\exists i,j, |\Delta_{i,j}|\geq B\right]\\
\leq~&\Pr \left[ \| Z \|_1 > \frac{4.0002}{\pi} n^2 \ln n  \right]+\Pr\left[\exists i,j, |\Delta_{i,j}|\geq B\right]\\
\leq~&\Pr \left[ \| Z \|_1 > \frac{4.0002}{\pi} n^2 \ln n  \right]+n^2\cdot 1/B\\
\leq~&\Pr \left[ \| Z \|_1 > 1.0001  \E[\|Z\|_1]\right]+n^2\cdot 1/B\\
\leq~&O(1/\log(n))+O(1/\log \log n)\\
\leq~&O(1/\log\log n)\\
\end{align*}
The second inequality follows by the definition of $Z$. The third inequality follows by the union bound and the cumulative distribution function of a half Cauchy random variable. The fourth inequality follows from $\E[\|Z\|_1]\leq n^2(1/\pi\cdot\ln(B^2+1)+1)\leq 4.0000001/\pi\cdot n^2\ln n$ when $n$ is sufficiently large.
\end{proof}

\subsection{``For Each'' Guarantee}

In the following Lemma, we show that, for each fixed coefficient vector $\alpha$, if the entry of $\alpha$ is too large, the fitting cost cannot be small.
\begin{lemma}[For each fixed $\alpha$, the entry cannot be too large]\label{lem:for_each_alpha_can_not_be_too_large}
Let $c>0$ be a sufficiently large constant, $n\geq d\geq 1,$ $u\in\mathbb{R}^n$ be any fixed vector and $\Delta \in \R^{n\times d}$ be a random matrix where $\forall i\in[n],j\in[d],\Delta_{i,j} \sim C(0,1)$ independently. For any fixed $\alpha \in \R^d$ with $\| \alpha \|_1 = n^{c}$,
\begin{align*}
\Pr_{\Delta \sim \{ C(0,1) \}^{n\times d} } [ \| ( u \cdot {\bf 1}^\top + \Delta ) \alpha \|_1 > n^3 ] > 1 - (1/n)^{\Theta(n)}.
\end{align*}
\end{lemma}
\begin{proof}
Let $c$ be a sufficiently large constant.
Let $\alpha\in\mathbb{R}^d$ with $\|\alpha\|_1=n^c$. Let $u\in\mathbb{R}^n$ be any fixed vector. Let $\Delta \in \R^{n\times d}$ be a random matrix where $\forall i\in[n],j\in[d],\Delta_{i,j} \sim C(0,1).$ Then $\Delta\alpha\in\mathbb{R}^n$ is a random vector with each entry drawn independently from $C(0,\|\alpha\|_1)$. Due to the probability density function of standard Cauchy random variables,
\begin{align*}
\Pr[\|\Delta\alpha\|_1<n^3]\geq \Pr[\|\Delta\alpha+u\cdot{\bf 1}^\top\alpha\|_1<n^3].
\end{align*}
It suffices to upper bound $\Pr[\|\Delta\alpha\|_1<n^3].$ If $c>10,$ then due to the cumulative distribution function of Cauchy random variables, for a fixed $i\in [n],$ $\Pr[|(\Delta\alpha)_i|<n^3]<1/n.$ Thus, $\Pr[\|\Delta\alpha\|_1<n^3]<(\frac{1}{n})^n.$ Thus,
\begin{align*}
\Pr_{\Delta \sim \{ C(0,1) \}^{n\times d} } [ \| ( u \cdot {\bf 1}^\top + \Delta ) \alpha \|_1 > n^3 ] > 1 - (1/n)^n.
\end{align*}
\end{proof}

\subsection{From ``For Each'' to ``For All'' via an $\epsilon$-Net}

\begin{definition}[$\epsilon$-net for the $\ell_1$-norm ball]\label{def:l1epsnet}
Let $A\in\mathbb{R}^{n\times d}$ have rank $d$, and let $L=\{y\in\mathbb{R}^n\mid y=Ax,x\in\mathbb{R}^d\}$ be the column space of $A.$ An $\epsilon$-net of the $\ell_1$-unit sphere $\mathcal{S}^{d-1}=\{y\mid \|y\|_1=1,y\in L\}\subset L$ is a set $N\subset \mathcal{S}^{d-1}$ of points for which $\forall y\in \mathcal{S}^{d-1},\exists y'\in N$ for which $\|y-y'\|\leq \epsilon.$
\end{definition}

\cite{ddhkm09} proved an upper bound on the size of an $\epsilon$-net.
\begin{lemma}[See, e.g., the ball $B$ on page 2068 of~\cite{ddhkm09}]\label{lem:l1_eps_net_size}
Let $A\in\mathbb{R}^{n\times d}$ have rank $d$, and let $L=\{y\in\mathbb{R}^n\mid y=Ax,x\in\mathbb{R}^d\}$ be the column space of $A.$ For $\epsilon\in(0,1),$ an $\epsilon$-net (Definition~\ref{def:l1epsnet}) $N$ of the $\ell_1$-unit sphere $\mathcal{S}^{d-1}=\{y\mid \|y\|_1=1,y\in L\}\subset L$ exists. Furthermore, the size of $N$ is at most $(3/\epsilon)^d.$
\end{lemma}

\begin{lemma}[For all possible $\alpha$, the entry cannot be too large]\label{lem:for_all_alpha_can_not_be_too_large}
 Let $n\geq 1,d=n^{o(1)}.$  Let $u = n^{c_0} \cdot {\bf 1} \in \R^n$ denote a fixed vector where $c_0$ is a constant. Let $\Delta \in \R^{n\times d}$ be a random matrix where $\forall i\in[n],j\in[d],\Delta_{i,j} \sim C(0,1)$ independently. Let $c>0$ be a sufficiently large constant. Conditioned on $\|\Delta\|_1\leq n^3,$ with probability at least $1-(1/n)^{\Theta(n)},$ for all $\alpha\in\mathbb{R}^d$ with $\|\alpha\|_1\geq n^c,$ we have $\| ( u \cdot {\bf 1}^\top + \Delta ) \alpha \|_1>0.9n^3.$
\end{lemma}
\begin{proof}
Due to Lemma~\ref{lem:l1_eps_net_size}, there is a set $N\subset\{\alpha\in\mathbb{R}^d\mid \|\alpha\|_1=n^c \}\subset \mathbb{R}^{d}$ with $|N|\leq 2^{\Theta(d\log n)}$ such that $\forall \alpha\in\mathbb{R}^d$ with $\|\alpha\|_1=n^c,$ $\exists \alpha'\in N$ such that $\|\alpha-\alpha'\|_1\leq 1/n^{c'}$ where $c'>c_0+100$ is a constant. By applying Lemma~\ref{lem:for_each_alpha_can_not_be_too_large} and union bounding over all the points in $N,$ with probability at least $1-(1/n)^n\cdot|N|\geq 1-(1/n)^n\cdot 2^{n^{o(1)}}=1-(1/n)^{\Theta(n)},$ $\forall \alpha'\in N,$ $\| ( u \cdot {\bf 1}^\top + \Delta ) \alpha' \|_1>n^3.$ $\forall \alpha\in\mathbb{R}^{d}$ with $\|\alpha\|_1=n^c,$ we can find $\alpha'\in N$ such that $\|\alpha-\alpha'\|_1\leq 1/n^{c'}.$ Let $\gamma=\alpha-\alpha'.$ Then,
\begin{align*}
&\|(u\cdot {\bf 1}^\top+\Delta)\alpha\|_1\\
=~& \|(u\cdot {\bf 1}^\top+\Delta)(\alpha'+\gamma)\|_1\\
\geq~ & \|(u\cdot {\bf 1}^\top+\Delta)\alpha'\|_1-\|(u\cdot {\bf 1}^\top+\Delta)\gamma\|_1\\
\geq~& n^3-\sqrt{n}\|(u\cdot {\bf 1}^\top+\Delta)\gamma\|_2\\
\geq~& n^3-\sqrt{n}(\|u\cdot {\bf 1}^\top\|_2+\|\Delta\|_2)\|\gamma\|_2\\
\geq~& n^3-n^{c_0+50}/n^{c'}\\
\geq~& 0.9n^3.
\end{align*}
The first equality follows from $\alpha=\alpha'+\gamma.$ The first inequality follows by the triangle inequality. The second inequality follows by the relaxation from the $\ell_1$ norm to the $\ell_2$ norm. The third inequality follows from the operator norm and the triangle inequality. The fourth inequality follows using $\|\Delta\|_2\leq\|\Delta\|_1\leq n^3,\|u\|_2\leq n^{c_0+10},\|\gamma\|_2\leq \|\gamma\|_1\leq (1/n)^{c'}.$ The last inequality follows since $c'>c_0+100.$

For $\alpha\in\mathbb{R}^n$ with $\|\alpha\|_1>n^c,$ let $\alpha'=\alpha/\|\alpha\|_1\cdot n^c.$ Then
\begin{align*}
\|(u\cdot {\bf 1}^\top+\Delta)\alpha\|_1\geq \|(u\cdot {\bf 1}^\top+\Delta)\alpha'\|_1 \geq 0.9n^3.
\end{align*}
\end{proof}

\subsection{Bounding the Cost from the Large-Entry Part via ``Bad'' Regions}

In this section, we will use the concept of \emph{well-conditioned basis} in our analysis.

\begin{definition}[Well-conditioned basis~\cite{ddhkm09}]\label{def:well_conditioned_basis_for_l1}
	Let $A\in\mathbb{R}^{n\times m}$ have rank $d$. Let $p\in[1,\infty),$ and let $\|\cdot\|_q$ be the dual norm of $\|\cdot\|_p,$ i.e., $1/p+1/q=1.$ If $U\in\mathbb{R}^{n\times d}$ satisfies
	\begin{enumerate}
		\item $\|U\|_p\leq \alpha,$
		\item $\forall z\in\mathbb{R}^d,\|z\|_q\leq \beta\|Uz\|_p,$
	\end{enumerate}
	then $U$ is an $(\alpha,\beta,p)$ well-conditioned basis for the column space of $A$.
\end{definition}

The following theorem gives an existence result of a well-conditioned basis. 

\begin{theorem}[$\ell_1$ well-conditioned basis~\cite{ddhkm09}]\label{thm:l1_well_conditioned_basis}
	Let $A\in\mathbb{R}^{n\times m}$ have rank $d$. There exists $U\in\mathbb{R}^{n\times d}$ such that $U$ is a $(d,1,1)$ well-conditioned basis for the column space of $A$.
\end{theorem}

In the following lemma, we consider vectors from low-dimensional subspaces. For a coordinate, if there is a vector from the subspace for which this entry is large, but the norm of the vector is small, then this kind of coordinate is pretty ``rare''. More formally, 
\begin{lemma}\label{lem:size_of_the_bad_region}
Given a matrix $U\in \R^{n\times r}$ for a sufficiently large $n\geq 1$, let $r=n^{o(1)}$. Let $S = \{ y | y = U x , x \in \R^r \}$. Let the set $T$ denote $\{ i\in [n] ~ | ~ \exists y \in S, |y_i| \geq n^{1.0001} \text{~and~} \| y \|_1 < 8 n^2 \ln n  \}$. Then we have
\begin{align*}
|T| \leq n^{0.99999 }.
\end{align*}
\end{lemma}
\begin{proof}
Due to Theorem~\ref{thm:l1_well_conditioned_basis}, let $U\in\mathbb{R}^{n\times r}$ be the $(r,1,1)$ well-conditioned basis of the column space of $U$. If $i\in T,$ then $\exists x\in \mathbb{R}^r$ such that $|(Ux)_i|\geq n^{1.0001}$ and $\|Ux\|_1<8n^2\ln n.$ Thus, we have
\begin{align*}
n^{1.0001}\leq |(Ux)_i|\leq \|U^i\|_1\|x\|_{\infty}\leq \|U^i\|_1\|Ux\|_1\leq \|U^i\|_1\cdot 8n^2\ln n.
\end{align*}
The first inequality follows using $n^{1.0001}\leq |(Ux)_i|.$ The second inequality follows by H\"{o}lder's inequality. The third inequality follows by the second property of the well-conditioned basis. The fourth inequality follows using $\|Ux\|_1<8n^2\ln n.$ Thus, we have
\begin{align*}
\|U^i\|_1\geq n^{1.0001}/n^{2+o(1)}\geq 1/n^{0.9999-o(1)}.
\end{align*}
Notice that $\sum_{j=1}^n \|U^j\|_1=\|U\|_1\leq r.$ Thus,
\begin{align*}
|T|\leq r/(1/n^{0.9999-o(1)})=n^{0.9999+o(1)}\leq n^{0.99999}.
\end{align*}
\end{proof}

\begin{definition}[Bad region]
Given a matrix $U \in \R^{n\times r}$, we say ${\cal B}(U) = \{ i \in [n] ~|~ \exists y \in \mathrm{colspan}(U) \subset \R^n \mathrm{~s.t.~} y_i \geq n^{1.0001} \mathrm{~and~} \| y \|_1 \leq 8 n^2 \ln n \}$ is a bad region for $U$.
\end{definition}

Next we state a lower and an upper bound on the probability that a Cauchy random variable is in a certain range,
\begin{claim}\label{cla:prob_Cauchy_range}
Let $X\sim C(0,1)$ be a standard Cauchy random variable. Then for any $x>1549,$
\begin{align*}
\frac{2}{\pi}\cdot\frac{\ln(1.001)}{x}\geq \Pr[|X|\in (x,1.001x]]\geq \frac{1.999}{\pi}\cdot\frac{\ln(1.001)}{x}.
\end{align*}
\end{claim}
\begin{proof}
When $x>1549,$ $\frac{2}{\pi}\cdot\frac{\ln(1.001)}{x}\geq\frac{2}{\pi}\cdot (\tan^{-1}(1.001x)-\tan^{-1}(x))\geq \frac{1.999}{\pi}\cdot\frac{\ln(1.001)}{x}.$
\end{proof}

We build a level set for the ``large'' noise values, and we show the bad region cannot cover much of the large noise. The reason is that the bad region is small, and for each row, there is always some large noise.

\begin{lemma}\label{lem:analysis_for_large_part}
Given a matrix $U\in \R^{n\times r}$ with $n$ sufficiently large, let $r=n^{o(1)}$, and consider a random matrix $\Delta \in \R^{n\times (n-r)}$ with $\Delta_{i,j} \sim C(0,1)$ independently. Let $L_t = \{ (i,j) ~|~ (i,j) \in [n]\times [n-r], |\Delta_{i,j}| \in (1.001^t , 1.001^{t+1} ] \}$. With probability at least $1-1/2^{n^{\Theta(1)}}$, for all $t \in ( \frac{1.0002\ln n}{\ln 1.001} , \frac{1.9999 \ln n}{\ln 1.001} ) \cap \mathbb{N}$,
\begin{align*}
| L_t \setminus ({\cal B}(U) \times [n-r]) | \geq n(n-r) \cdot 1.998\cdot \ln(1.001)/(\pi\cdot 1.001^t).
\end{align*}
\end{lemma}
\begin{proof}
Let $N=n\cdot(n-r).$ Then according to Claim~\ref{cla:prob_Cauchy_range}, $\forall t\in ( \frac{1.0002\ln n}{\ln 1.001} , \frac{1.9999 \ln n}{\ln 1.001} )\cap \mathbb{N},$ $\E(|L_t|)\geq N\cdot 1.999\cdot \ln(1.001)/(\pi\cdot 1.001^t)\geq n^{\Theta(1)}.$ For a fixed $t,$ by a Chernoff bound, with probability at least $1-1/2^{n^{\Theta(1)}}$, $|L_t|\geq N\cdot 1.9989\cdot \ln(1.001)/(\pi\cdot 1.001^t).$ Due to Lemma~\ref{lem:size_of_the_bad_region}, $|{\cal B}(U) \times [n-r]|\leq n^{0.99999}(n-r)=N/n^{0.00001}.$ Due to the Chernoff bound, with probability at least $1-1/2^{n^{\Theta(1)}},$ $| L_t \cap ({\cal B}(U) \times [n-r]) |<N/n^{0.00001}\cdot 2.0001\cdot \ln(1.001)/(\pi\cdot 1.001^t).$ Thus, with probability at least $1-1/2^{n^{\Theta(1)}},$ $|L_t\setminus ({\cal B}(U) \times [n-r])|\geq N\cdot 1.998\cdot \ln(1.001)/(\pi\cdot 1.001^t).$ By taking a union bound over all $t\in ( \frac{1.0002\ln n}{\ln 1.001} , \frac{1.9999 \ln n}{\ln 1.001} ) \cap \mathbb{N}$, we complete the proof.
\end{proof}

\begin{lemma}[The cost of the large noise part]\label{lem:bound_for_large_part}
Let $n\geq 1$ be sufficiently large, and let $r=n^{o(1)}.$ Given a matrix $U\in \R^{n\times r}$, and a random matrix $\Delta \in \R^{n\times (n-r)}$ with $\Delta_{i,j} \sim C(0,1)$ independently, let ${\cal I}= \{ (i,j) \in [n]\times[n-r] ~ | ~ |\Delta_{i,j}| \geq n^{1.0002} \}$. If $\| \Delta \|_1 \leq 4 n^2\ln n$, then with probability at least $1-1/2^{n^{\Theta(1)}}$, for all $X\in \R^{r\times n}$, either
\begin{align*}
\sum_{(i,j)\in {\cal I}} | (U X - \Delta)_{i,j} | >  \frac{1.996}{\pi} n^2 \ln n,
\end{align*}
or
\begin{align*}
\|UX-\Delta\|_1>4n^2\ln n
\end{align*}
\end{lemma}
\begin{proof}
\begin{align}
&\sum_{(i,j)\in {\cal I}} | (U X - \Delta)_{i,j} | \notag\\
\geq~& \sum_{(i,j)\in {\cal I}\setminus \mathcal{B}(U)} | (U X - \Delta)_{i,j} |\notag\\
\geq~&\sum_{(i,j)\in {\cal I}\setminus \mathcal{B}(U)} | (\Delta)_{i,j} |-\sum_{(i,j)\in {\cal I}\setminus \mathcal{B}(U)} | (UX)_{i,j} |\label{eq:large_part_fixed_and_random}
\end{align}
Let $N=n(n-r).$ By a Chernoff bound and the cumulative distribution function of a Cauchy random variable, with probability at least $1-1/2^{n^{\Theta(1)}},$ $|{\cal I}|\leq 1.1\cdot N/n^{1.0002}.$ If $\exists (i,j)\in {\cal I}\setminus \mathcal{B}(U)$ which has $| (UX)_{i,j} |>n^{1.0001}$, then according to the definition of $\mathcal{B}(U),$ $\|UX\|_1\geq \|(UX)_j\|_1\geq 8n^2\ln n.$ Due to the triangle inequality, $\|UX-\Delta\|_1\geq \|UX\|_1-\|\Delta\|_1\geq 4n^2\ln n.$ If $\forall (i,j)\in {\cal I}\setminus \mathcal{B}(U)$ we have $| (UX)_{i,j} |\leq n^{1.0001},$ then
\begin{align}\label{eq:large_part_fixed_part}
\sum_{(i,j)\in {\cal I}\setminus \mathcal{B}(U)} | (UX)_{i,j} |\leq |\mathcal{I}|\cdot n^{1.0001}\leq 1.1\cdot N/n^{0.0001}.
\end{align}

Due to Lemma~\ref{lem:analysis_for_large_part}, with probability at least $1-1/2^{n^{\Theta(1)}},$
\begin{align}
&\sum_{(i,j)\in {\cal I}\setminus \mathcal{B}(U)} | (\Delta)_{i,j} |\notag\\
\geq~&\sum_{t\in( \frac{1.0002\ln n}{\ln 1.001} , \frac{1.9999 \ln n}{\ln 1.001} ) \cap \mathbb{N}} 1.001^t \cdot N \cdot 1.998\cdot \ln(1.001)/(\pi\cdot 1.001^t)\notag\\
\geq~&\frac{1.997}{\pi}\cdot N \ln n. \label{eq:large_part_random_part}
\end{align}
We plug~\eqref{eq:large_part_fixed_part} and~\eqref{eq:large_part_random_part} into~\eqref{eq:large_part_fixed_and_random}, from which we have
\begin{align*}
\sum_{(i,j)\in {\cal I}} | (U X - \Delta)_{i,j} |\geq \frac{1.996}{\pi} n^2\ln n.
\end{align*}
\end{proof}

\subsection{Cost from the Sign-Agreement Part of the Small-Entry Part}

We use $-y$ to fit $\Delta$ (we think of $A_S \alpha = A_S^* \alpha - y$, and want to minimize $\| - y - \Delta \|_1$). If the sign of $y_j$ is the same as the sign of $\Delta_j$, then both coordinate values will collectively contribute.
\begin{lemma}[The contribution from $\Delta_i$ when $\Delta_i$ and $y_i$ have the same sign]\label{lem:small_part_same_sign_delta}
Suppose we are given a vector $y\in \R^n$ and a random vector $\Delta \in \R^{n}$ with $\Delta_{j} \sim C(0,1)$ independently. Then with probability at least $1-1/2^{n^{\Theta(1)}},$
\begin{align*}
\sum_{j ~:~ \sign(y_j)=\sign(\Delta_{j}) \mathrm{~and~} |\Delta_{j}| \leq n^{0.9999} } |\Delta_{j}| > \frac{0.9998}{\pi} n \ln n.
\end{align*}
\end{lemma}
\begin{proof}
For $j\in[n]$, define the random variable
\begin{align*}
Z_j=\left\{\begin{array}{ll}\Delta_j&0<\Delta_j\leq n^{0.9999}\\0&\text{otherwise}\end{array}\right..
\end{align*}
Then, we have
\begin{align*}
\Pr\left[\sum_{j ~:~ \sign(y_j)=\sign(\Delta_{i,j}) \mathrm{~and~} |\Delta_{j}| \leq n^{0.9999} } |\Delta_{j}| > \frac{0.9998}{\pi} n \ln n\right]
=\Pr\left[\sum_{j=1}^n Z_j>\frac{0.9998}{\pi} n \ln n\right].
\end{align*}
Let $B=n^{0.9999}.$
For $j\in [n],$
\begin{align*}
\E[Z_j]=\frac{1}{\pi}\int_{0}^{B} \frac{x}{1+x^2} \mathrm{d}x=\frac{1}{2\pi}\ln(B^2+1).
\end{align*}
Also, 
\begin{align*}
\E[Z_j^2]=\frac{1}{\pi}\int_{0}^{B} \frac{x^2}{1+x^2} \mathrm{d}x=\frac{B-\tan^{-1}(B)}{\pi}\leq B.
\end{align*}
By Bernstein's inequality,
\begin{align*}
&\Pr\left[\E\left[\sum_{j=1}^n Z_j\right]-\sum_{j=1}^n Z_j>10^{-5}\E\left[\sum_{j=1}^n Z_j\right]\right]\\
\leq~&\exp\left(-\frac{0.5\cdot \left(10^{-5}\E\left[\sum_{j=1}^n Z_j\right]\right)^2}{\sum_{j=1}^n \E[Z_j^2]+\frac{1}{3}B\cdot 10^{-5}\E\left[\sum_{j=1}^n Z_j\right]}\right)\\
\leq~&\exp\left(-\frac{5\cdot 10^{-11}n^2\ln^2(B^2+1)/(4\pi^2)}{nB+\frac{1}{3}B\cdot 10^{-5}n\ln(B^2+1)/(2\pi)}\right)\\
\leq~&e^{-n^{\Theta(1)}}.\\
\end{align*}
The last inequality follows since $B=n^{0.9999}.$ Thus, we have
\begin{align*}
\Pr\left[\sum_{j=1}^n Z_j<0.9998/\pi\cdot n\ln n\right]\leq\Pr\left[\sum_{j=1}^n Z_j<0.99999n\ln(B^2+1)/(2\pi)\right]\leq e^{-n^{\Theta(1)}}.
\end{align*}
\end{proof}

\begin{lemma}[Bound on level sets of a Cauchy vector]\label{lem:bound_of_Cauchy_level_size}
Suppose we are given a random vector $y\in \R^n$ with $y_i\sim C(0,1)$ chosen independently. Let
\begin{align*}
L_t^- = \{ i \in [n] ~ | ~ -y_i \in (1.001^t, 1.001^{t+1}] \} \mathrm{~and~}  L_t^+ = \{ i \in [n] ~ | ~ y_i \in (1.001^t, 1.001^{t+1}] \} .
\end{align*}
With probability at least $1-1/2^{n^{\Theta(1)}}$, for all $t\in(\frac{\ln 1549}{\ln 1.001},\frac{0.9999\ln n}{\ln 1.001})\cap \mathbb{N}$,
\begin{align*}
\min ( |L_t^-|, |L_t^+| ) \geq 0.999n \cdot \frac{1}{\pi}  \frac{\ln 1.001}{1.001^t} .
\end{align*}
\end{lemma}
\begin{proof}
For $i\in[n],t\geq \frac{\ln 1549}{\ln 1.001},$ according to Claim~\ref{cla:prob_Cauchy_range}, $\Pr[y_i\in (1.001^t,1.001^{t+1}]]\geq 0.9995/\pi\cdot \ln(1.0001)/1.001^t.$ Thus, $\E[|L_t^+|]=\E[|L_t^-|]=n\cdot 0.9995/\pi\cdot \ln(1.0001)/1.001^t.$ Since $t\leq \frac{0.9999\ln n}{\ln 1.001},$ $1.001^t\leq n^{0.9999},$ we have $\E[|L_t^+|]=\E[|L_t^-|]\geq n^{\Theta(1)}.$ By applying a Chernoff bound,
\begin{align*}
\Pr[|L_t^+|>0.999n/\pi\cdot \ln(1.0001)/1.001^t]\geq 1-1/2^{n^{\Theta(1)}}.
\end{align*}
Similarly, we have
\begin{align*}
\Pr[|L_t^-|>0.999n/\pi\cdot \ln(1.0001)/1.001^t]\geq 1-1/2^{n^{\Theta(1)}}.
\end{align*}
By taking a union bound over all the $L_t^+$ and $L_t^-,$ we complete the proof.
\end{proof}

\begin{lemma}[The contribution from $y_i$ when $\Delta_i$ and $y_i$ have the same sign]\label{lem:small_part_same_sign_alpha}
Let $u=\eta\cdot\mathbf{1}\in\mathbb{R}^n$ where $\eta\in\mathbb{R}$ is an arbitrary real number. Let $y\in \R^n$ be a random vector with $y_i\sim C(0,\beta)$ independently for some $\beta>0.$ Let $\Delta\in\mathbb{R}^n$ be a random vector with $\Delta_i \sim C(0,1)$ independently. With probability at least $1-1/2^{n^{\Theta(1)}},$
\begin{align*}
\sum_{i ~: ~ \sign((u+y)_i) = \sign(\Delta_{i}) \mathrm{~and~} |\Delta_{i}| \leq n^{0.9999}} |(u+y)_i| \geq  \beta\cdot \frac{0.997}{\pi} n\ln n.
\end{align*}
\end{lemma}
\begin{proof}
For all $t\in(\frac{\ln 1549}{\ln 1.001},\frac{0.9999\ln n}{\ln 1.001})\cap \mathbb{N},$ define
\begin{align*}
L_t^- = \{ i \in [n] ~ | ~ -y_i \in (\beta\cdot 1.001^t, \beta\cdot 1.001^{t+1}] \} \mathrm{~and~}  L_t^+ = \{ i \in [n] ~ | ~ y_i \in (\beta\cdot1.001^t, \beta\cdot1.001^{t+1}] \} .
\end{align*}
Define
\begin{align*}
G=\{i\in[n]\mid \sign((u+y)_i) = \sign(\Delta_{i})\mathrm{~and~}|\Delta_{i}| \leq n^{0.9999}\}.
\end{align*}
Then $\forall i\in[n],\Pr[i\in G]\geq 0.5-1/n^{0.9999}\geq 0.4999999999.$
Due to Lemma~\ref{lem:bound_of_Cauchy_level_size},
\begin{align*}
\min ( |L_t^-|, |L_t^+| ) \geq 0.999n \cdot \frac{1}{\pi}  \frac{\ln 1.001}{1.001^t} \geq n^{\Theta(1)}.
\end{align*}
By a Chernoff bound and a union bound, with probability at least $1-1/2^{n^{\Theta(1)}},$ $\forall t\in(\frac{\ln 1549}{\ln 1.001},\frac{0.9999\ln n}{\ln 1.001})\cap \mathbb{N},$
\begin{align}
&\min ( |L_t^-\cap G|, |L_t^+ \cap G|)\notag\\
\geq~&0.499n \cdot \frac{1}{\pi}  \frac{\ln 1.001}{1.001^t}. \label{eq:size_of_pm_level}
\end{align}
Then we have
\begin{align*}
&\sum_{i\in G} |(u+y)_i|\\
\geq~&\sum_{t\in(\frac{\ln 1549}{\ln 1.001},\frac{0.9999\ln n}{\ln 1.001})\cap \mathbb{N}}\left(\sum_{i\in L_t^+,i\in G}|y_i+\eta|+\sum_{i\in L_t^-,i\in G}|-y_i-\eta|\right)\\
\geq~&\sum_{t\in(\frac{\ln 1549}{\ln 1.001},\frac{0.9999\ln n}{\ln 1.001})\cap \mathbb{N}}0.499n \cdot \frac{1}{\pi}  \frac{\ln 1.001}{1.001^t}\cdot 2\cdot 1.001^t\cdot \beta\\
\geq~&\beta\cdot\frac{0.997}{\pi}n\ln n
\end{align*}
The second inequality follows by Equation~\eqref{eq:size_of_pm_level} and the triangle inequality, i.e., $\forall a,b,c\in\mathbb{R},|a+c|+|b-c|\geq |a+b|.$
\end{proof}

\subsection{Cost from the Sign-Disagreement Part of the Small-Entry Part}

\begin{lemma}\label{lem:small_part_different_sign}
Given a vector $y\in \R^n$ and a random vector $\Delta\in \R^{n}$ with $\Delta_{i} \sim C(0,1)$ independently, with probability at least $1-1/2^{n^{\Theta(1)}}$,
\begin{align*}
\sum_{i ~ : ~ \sign(y_i) \neq \sign(\Delta_{i}) \mathrm{~and~} |\Delta_{i}| < n^{0.9999} } |y_i + \Delta_{i}| > \frac{0.03}{\pi}n\ln n.
\end{align*}
\end{lemma}
\begin{proof}
For $t\in [0,\frac{0.9999\ln n}{\ln 4})\cap \mathbb{N}$ define
\begin{align*}
L_t=\{i\in[n]\mid \sign(y_i) \neq \sign(\Delta_{i}),|\Delta_i|\in (4^t,4^{t+1}],|\Delta_i|\not\in[|y_i|-4^t,|y_i|+4^t]\}.
\end{align*}
$\forall x\geq 1,y>0,$ we have
\begin{align*}
&\Pr_{X\sim C(0,1)}[|X|\in (x,4x],|X|\not\in [y-x,y+x]]\\
\geq~& \Pr_{X\sim C(0,1)}[|X|\in (3x,4x]]\\
=~ &\frac{2}{\pi}\cdot(\tan^{-1}(4x)-\tan^{-1}(3x))\\
\geq~& \frac{0.1}{\pi}\cdot\frac{\ln(4)}{x}
\end{align*}
Thus, $\forall i\in[n],t\in [0,\frac{0.9999\ln n}{\ln 4})\cap \mathbb{N},$
\begin{align*}
\Pr[i\in L_t]\geq \frac{0.05}{\pi}\cdot\frac{\ln(4)}{4^t}.
\end{align*}
Thus, $\forall t\in [0,\frac{0.9999\ln n}{\ln 4})\cap \mathbb{N},\E[|L_t|]\geq 0.05n/\pi\cdot \ln(4)/4^t\geq n^{\Theta(1)}.$ By a Chernoff bound and a union bound, with probability at least $1-1/2^{n^{\Theta(1)}}$ $\forall t\in [0,\frac{0.9999\ln n}{\ln 4})\cap \mathbb{N},$ $|L_t|\geq 0.04n/\pi\cdot \ln(4)/4^t.$ Thus, we have, with probability at least $1-1/2^{n^{\Theta(1)}},$
\begin{align*}
&\sum_{i ~ : ~ \sign(y_i) \neq \sign(\Delta_{i}) \mathrm{~and~} |\Delta_{i}| < n^{0.9999} } |y_i + \Delta_{i}|\\
\geq~&\sum_{t\in [0,\frac{0.9999\ln n}{\ln 4})\cap \mathbb{N}}|L_t|\cdot 4^t\\
\geq~&\frac{0.03}{\pi}n\ln n.
\end{align*}
\end{proof}

\subsection{Overall Cost of the Small-Entry Part}

\begin{lemma}[For each]\label{lem:fixed_small_part}
Let $u=\eta\cdot \mathbf{1}\in\mathbb{R}^n$ where $\eta\in\mathbb{R}$ is an arbitrary real number. Let $\alpha\in\mathbb{R}^d$ where $\|\alpha\|_1\geq 1-10^{-20}.$ Let $\Delta\in\mathbb{R}^{n\times(d+1)}$ and $\forall (i,j)\in [n]\times [d+1],\Delta_{i,j}\sim C(0,1)$ are i.i.d. standard Cauchy random variables. Then with probability at least $1-1/2^{n^{\Theta(1)}},$
\begin{align*}
\sum_{j\in[n],|\Delta_{j,d+1}|<n^{0.9999}}|(u+\Delta_{d+1}-(u\mathbf{1}^\top+\Delta_{[d]})\alpha)_j|\geq \frac{2.025}{\pi}n\ln n.
\end{align*}
\end{lemma}
\begin{proof}
Let $G_1=\{j\in[n]\mid |\Delta_{j,d+1}|<n^{0.9999},\sign((u(1-\mathbf{1}^\top\alpha)-\Delta_{[d]}\alpha)_j)=\sign(\Delta_{d+1})_j)\},G_2=\{j\in[n]\mid |\Delta_{j,d+1}|<n^{0.9999},\sign((u(1-\mathbf{1}^\top\alpha)-\Delta_{[d]}\alpha)_j)\not=\sign(\Delta_{d+1})_j)\}.$ Notice that $\Delta_{[d]}\alpha$ is a random vector with each entry independently drawn from $C(0,\|\alpha\|_1).$
Then with probability at least $1-1/2^{n^{\Theta(1)}},$
\begin{align*}
&\sum_{j\in[n],|\Delta_{j,d+1}|<n^{0.9999}}|(u+\Delta_{d+1}-(u\mathbf{1}^\top+\Delta_{[d]})\alpha)_j|\\
=~&\sum_{j\in[n],|\Delta_{j,d+1}|<n^{0.9999}} |(u(1-\mathbf{1}^\top\alpha)-\Delta_{[d]}\alpha+\Delta_{d+1})_j|\\
=~&\sum_{j\in G_1} |(u(1-\mathbf{1}^\top\alpha)-\Delta_{[d]}\alpha+\Delta_{d+1})_j|+\sum_{j\in G_2} |(u(1-\mathbf{1}^\top\alpha)-\Delta_{[d]}\alpha+\Delta_{d+1})_j|\\
=~& \sum_{j\in G_1} |(u(1-\mathbf{1}^\top\alpha)-\Delta_{[d]}\alpha)_j|+\sum_{j\in G_1}|(\Delta_{d+1})_j|+\sum_{j\in G_2} |(u(1-\mathbf{1}^\top\alpha)-\Delta_{[d]}\alpha+\Delta_{d+1})_j|\\
\geq~&\|\alpha\|_1\cdot\frac{0.997}{\pi}\cdot n\ln n+\frac{0.9998}{\pi}n\ln n+\frac{0.03}{\pi}n\ln n\\
\geq~&\frac{2.025}{\pi}n\ln n
\end{align*}
The first inequality follows by Lemma~\ref{lem:small_part_same_sign_alpha}, Lemma~\ref{lem:small_part_same_sign_delta} and Lemma~\ref{lem:small_part_different_sign}. The second inequality follows by $\|\alpha\|_1\geq 1-10^{-20}.$
\end{proof}

\begin{lemma}[For all]\label{lem:bound_for_small_part}
Let $c>0,c_0>0$ be two arbitrary constants. Let $u=\eta\cdot \mathbf{1}\in\mathbb{R}^n$ where $\eta\in\mathbb{R}$ satisfies $|\eta|\leq n^{c_0}$. Consider a random matrix $\Delta\in\mathbb{R}^{n\times(d+1)}$ with $d=n^{o(1)}$ and $\forall (i,j)\in [n]\times [d+1],\Delta_{i,j}\sim C(0,1)$ are i.i.d. standard Cauchy random variables. Conditioned on $\|\Delta\|_1\leq n^3,$ with probability at least $1-1/2^{n^{\Theta(1)}},$ $\forall \alpha\in\mathbb{R}^{d}$ with $1-10^{-20}\leq\|\alpha\|_1\leq n^c,$
\begin{align*}
\sum_{j\in[n],|\Delta_{j,d+1}|<n^{0.9999}}|(u+\Delta_{d+1}-(u\mathbf{1}^\top+\Delta_{[d]})\alpha)_j|\geq \frac{2.024}{\pi}n\ln n.
\end{align*}
\end{lemma}
\begin{proof}
Let $\mathcal{N}$ be a set of points:
\begin{align*}
\mathcal{N}=\left\{\alpha\in\mathbb{R}^d\mid 1-10^{-20}\leq \|\alpha\|_1\leq n^c\mathrm{~and~} \exists q\in \mathbb{Z}^d,\mathrm{~such~that~}\alpha=q/n^{c+c_0+1000}\right\}.
\end{align*}
Since $d=n^{o(1)},$ we have $|\mathcal{N}|\leq (n^{2c+c_0+2000})^d=2^{n^{o(1)}}.$ By Lemma~\ref{lem:fixed_small_part} and a union bound, with probability at least $1-1/2^{n^{\Theta(1)}}\cdot |\mathcal{N}|\geq 1-1/2^{n^{\Theta(1)}},$ $\forall \alpha\in \mathcal{N},$ we have
\begin{align*}
\sum_{j\in[n],|\Delta_{j,d+1}|<n^{0.9999}}|(u+\Delta_{d+1}-(u\mathbf{1}^\top+\Delta_{[d]})\alpha)_j|\geq \frac{2.025}{\pi}n\ln n.
\end{align*}
Due to the construction of $\mathcal{N},$ we have $\forall \alpha\in\mathbb{R}^{d}$ with $1-10^{-20}\leq\|\alpha\|_1\leq n^c,$ $\exists \alpha'\in\mathcal{N}$ such that $\|\alpha-\alpha'\|_\infty\leq 1/n^{c+c_0+1000}.$ Let $\gamma=\alpha-\alpha'.$ Then
\begin{align*}
&\sum_{j\in[n],|\Delta_{j,d+1}|<n^{0.9999}}|(u+\Delta_{d+1}-(u\mathbf{1}^\top+\Delta_{[d]})\alpha)_j|\\
=~&\sum_{j\in[n],|\Delta_{j,d+1}|<n^{0.9999}}|(u+\Delta_{d+1}-(u\mathbf{1}^\top+\Delta_{[d]})(\alpha'+\gamma))_j|\\
\geq~&\sum_{j\in[n],|\Delta_{j,d+1}|<n^{0.9999}}|(u+\Delta_{d+1}-(u\mathbf{1}^\top+\Delta_{[d]})\alpha')_j|-\sum_{j\in[n],|\Delta_{j,d+1}|<n^{0.9999}}|((u\mathbf{1}^\top+\Delta_{[d]})\gamma)_j|\\
\geq~&\sum_{j\in[n],|\Delta_{j,d+1}|<n^{0.9999}}|(u+\Delta_{d+1}-(u\mathbf{1}^\top+\Delta_{[d]})\alpha')_j|-\|(u\mathbf{1}^\top+\Delta_{[d]})\gamma\|_1\\
\geq~&\frac{2.025}{\pi}n\ln n-1/n^{500}\\
\geq~&\frac{2.024}{\pi}n\ln n
\end{align*}
The first equality follows from $\alpha=\alpha'+\gamma.$ The first inequality follows by the triangle inequality. The third inequality follows from $\|\gamma\|_1\leq 1/n^{c+c_0+800},\|u\mathbf{1}^\top\|_1\leq n^{c_0+10},\|\Delta\|_1\leq n^3,$ and $\forall \alpha'\in \mathcal{N},$
\begin{align*}
\sum_{j\in[n],|\Delta_{j,d+1}|<n^{0.9999}}|(u+\Delta_{d+1}-(u\mathbf{1}^\top+\Delta_{[d]})\alpha')_j|\geq \frac{2.025}{\pi}n\ln n.
\end{align*}
\end{proof}

\subsection{Main result}

\begin{theorem}[Formal version of Theorem~\ref{thm:intro_l1_hardness}]\label{thm:dis_l1_hardness}
Let $n>0$ be sufficiently large. Let $A=\eta\cdot \mathbf{1}\cdot\mathbf{1}^\top+\Delta\in\mathbb{R}^{n\times n}$ be a random matrix where $\eta=n^{c_0}$ for some sufficiently large constant $c_0,$ and $\forall i,j\in[n],\Delta_{i,j}\sim C(0,1)$ are i.i.d. standard Cauchy random variables. Let $r=n^{o(1)}.$ Then with probability at least $1-O(1/\log\log n),$ $\forall S\subset[n]$ with $|S|=r,$
\begin{align*}
\min_{X\in\mathbb{R}^{r\times n}} \|A_SX-A\|_1\geq 1.002\|\Delta\|_1
\end{align*}
\end{theorem}
\begin{proof}
We first argue that for a fixed set $S,$ conditioned on $\|\Delta\|_1\leq 100n^2\ln n,$ with probability at least $1-1/2^{n^{\Theta(1)}},$
\begin{align*}
\min_{X\in\mathbb{R}^{r\times n}} \|A_SX-A\|_1\geq 1.002\|\Delta\|_1.
\end{align*}
Then we can take a union bound over the at most $n^r=2^{n^{o(1)}}$ possible choices of $S.$ It suffices to show for a fixed set $S,$  $\min_{X\in\mathbb{R}^{r\times n}} \|A_SX-A\|_1$ is not small.

Without loss of generality, let $S=[r],$ and we want to argue the cost
\begin{align*}
\min_{X\in\mathbb{R}^{r\times n}} \|A_SX-A\|_1\geq\min_{X\in\mathbb{R}^{r\times n}} \|A_SX_{[n]\setminus S}-A_{[n]\setminus S}\|_1\geq 1.002\|\Delta\|_1.
\end{align*}
Due to Lemma~\ref{lem:bound_of_Delta}, with probability at least $1-O(1/\log\log n),$ $\|\Delta\|_1\leq 4.0002/\pi\cdot n^2\ln n.$ Now, we can condition on $\|\Delta\|_1\leq 4.0002/\pi\cdot n^2\ln n.$

Consider $j\in[n]\setminus S.$ Due to Lemma~\ref{lem:for_all_alpha_can_not_be_too_large}, with probability at least $1-(1/n)^{\Theta(n)},$ for all $X_j\in\mathbb{R}^r$ with $\|X_j\|_1\geq n^c$ for some constant $c>0,$ we have
\begin{align*}
\|A_SX_j-A_j\|_1&=\|(\eta\cdot\mathbf{1}\cdot\mathbf{1}^\top+[\Delta_S\ \Delta_j])[X_j^\top\ -1]^\top\|_1\geq 0.9n^3.
\end{align*}

By taking a union bound over all $j\in[n]\setminus S,$ with probability at least $1-(1/n)^{\Theta(n)},$ for all $X\in\mathbb{R}^{r\times n}$ with $\exists j\in[n]\setminus S,\|X_j\|_1\geq n^c,$ we have
\begin{align*}
\|A_SX-A\|_1\geq 0.9n^3.
\end{align*}
Thus, we only need to consider the case $\forall j\in [n]\setminus S,\|X_j\|_1\leq n^c.$ Notice that we condition on $\|\Delta\|_1\leq 4.0002/\pi\cdot n^2\ln n.$ By Fact~\ref{fac:bound_of_sum_of_alpha}, we have that if $\|X_j\|_1\leq n^c$ and $|1-\mathbf{1}^\top X_j|>1-10^{-20},$ then $\|A_SX-A\|_1\geq \|A_SX_j-A_j\|_1> n^3.$

Thus, we only need to consider the case $\forall j\in [n]\setminus S,\|X_j\|_1\leq n^c,|1-\mathbf{1}^\top X_j|\leq 1-10^{-20}.$ $\forall X\in\mathbb{R}^{r\times n}$ with $\forall j\in[n]\setminus S, \|X_j\|_1\leq n^c,|1-\mathbf{1}^\top X_j|\leq 1-10^{-20},$ if $\|A_SX_{[n]\setminus S}-A_{[n]\setminus S}\|_1\leq 4n^2\ln n,$ then
\begin{align*}
&\|A_SX_{[n]\setminus S}-A_{[n]\setminus S}\|_1\\
=~&\|(\eta\cdot\mathbf{1}\cdot\mathbf{1}^\top+\Delta_S)X_{[n]\setminus S}-(\eta\cdot\mathbf{1}\cdot\mathbf{1}^\top+\Delta_{[n]\setminus S})\|_1\\
\geq~& \sum_{i\in[n],j\in[n]\setminus S,|\Delta_{i,j}|\geq n^{1.0002}} |(((\eta\cdot\mathbf{1}\cdot\mathbf{1}^\top+\Delta_S)X_{[n]\setminus S}-\eta\cdot\mathbf{1}\cdot\mathbf{1}^\top)-\Delta_{[n]\setminus S})_{i,j}|\\
&+\sum_{i\in[n],j\in[n]\setminus S,|\Delta_{i,j}|< n^{0.9999}} |((\eta\cdot\mathbf{1}\cdot\mathbf{1}^\top+\Delta_S)X_{[n]\setminus S}-(\eta\cdot\mathbf{1}\cdot\mathbf{1}^\top+\Delta_{[n]\setminus S}))_{i,j}|\\
\geq~& \frac{1.996}{\pi}\cdot n^2\ln n+\sum_{i\in[n],j\in[n]\setminus S,|\Delta_{i,j}|< n^{0.9999}} |((\eta\cdot\mathbf{1}\cdot\mathbf{1}^\top+\Delta_S)X_{[n]\setminus S}-(\eta\cdot\mathbf{1}\cdot\mathbf{1}^\top+\Delta_{[n]\setminus S}))_{i,j}|\\
=~&\frac{1.996}{\pi}\cdot n^2\ln n+\sum_{j\in[n]\setminus S}\sum_{i\in[n],|\Delta_{i,j}|< n^{0.9999}} |((\eta\cdot\mathbf{1}\cdot\mathbf{1}^\top+\Delta_S)X_j-\eta\cdot\mathbf{1}-\Delta_j)_i|\\
\geq~&\frac{1.996}{\pi}\cdot n^2\ln n+\sum_{j\in[n]\setminus S}\frac{2.024}{\pi}n\ln n\\
\geq~&\frac{1.996}{\pi}\cdot n^2\ln n+\frac{2.023}{\pi}n^2\ln n\\
\geq~&\frac{4.01}{\pi}\cdot n^2\ln n
\end{align*}
holds with probability at least $1-1/2^{n^\Theta(1)}.$ The first equality follows by the definition of $A$. The first inequality follows by the partition by $|\Delta_{i,j}|.$ Notice that $[\mathbf{1}\ \Delta_S]$ has rank at most $r+1=n^{o(1)}.$ Then, due to Lemma~\ref{lem:bound_for_large_part}, and the condition $\|A_SX_{[n]\setminus S}-A_{[n]\setminus S}\|_1\leq 4n^2\ln n,$ the second inequality holds with probability at least $1-1/2^{n^{\Theta(1)}}.$ The second equality follows by grouping the cost by each column. The third inequality holds with probability at least $1-1/2^{n^{\Theta(1)}}$ by Lemma~\ref{lem:bound_for_small_part}, and a union bound over all the columns in $[n]\setminus S.$ The fourth inequality follows by $n-r=n-n^{o(1)}\geq (1-10^{-100})n.$

Thus, conditioned on $\|\Delta\|_1\leq 4.0002/\pi\cdot n^2\ln n,$ with probability at least $1-1/2^{n^{\Theta(1)}},$ we have $\min_{X\in\mathbb{R}^{r\times n}} \|A_SX-A\|_1\geq \frac{4.02}{\pi}\cdot n^2\ln n.$ By taking a union bound over all the ${n\choose r} = 2^{n^{o(1)}}$ choices of $S,$ we have that conditioned on $\|\Delta\|_1\leq \frac{4.0002}{\pi}n^2\ln n,$ with probability at least $1-1/2^{n^{\Theta(1)}},$ $\forall S\subset [n]$ with $|S|=r=n^{o(1)},$
$\min_{X\in\mathbb{R}^{r\times n}} \|A_SX-A\|_1\geq \frac{4.02}{\pi}\cdot n^2\ln n.$ Since $4.01/4.0002>1.002,$
\begin{align*}
\min_{X\in\mathbb{R}^{r\times n}} \|A_SX-A\|_1\geq 1.002 \|\Delta\|_1.
\end{align*}
Since $\|\Delta\|_1\leq \frac{4.0002}{\pi}n^2\ln n$ happens with probability at least $1-O(1/\log \log n),$ this completes the proof.
\end{proof}

\end{document}